\documentclass[12pt,reqno]{amsart}

\usepackage{a4wide}
\usepackage{graphics}
\usepackage{geometry}
\usepackage{amssymb, amsmath}
\usepackage[applemac]{inputenc}
\usepackage[sans]{dsfont}

\newtheorem{theorem}{Theorem}[section]
\newtheorem{lemma}[theorem]{Lemma}
\newtheorem{proposition}[theorem]{Proposition}
\newtheorem{corollary}[theorem]{Corollary}

\numberwithin{equation}{section}

\renewcommand{\Im}{\mathrm{Im}}
\renewcommand{\Re}{\mathrm{Re}}

\renewcommand{\d}{\mathrm{d}}
\renewcommand{\i}{\mathrm{i}}
\newcommand{\DETAILS}[1]{}

\newcommand{\DATUM}{November 4, 2009}              
\begin{document}

\title{Local Decay in Non-relativistic QED}
\author{T. Chen}
\address[T. Chen]{Department of Mathematics, University of Texas at Austin, USA}
\email{tc@math.utexas.edu}
\author{J. Faupin}
\address[J. Faupin]{Institut de Math{\'e}matiques de Bordeaux \\
UMR-CNRS 5251, Universit{\'e} de Bordeaux 1 \\
351 cours de la lib{\'e}ration, 33405 Talence Cedex, France}
 \email{jeremy.faupin@math.u-bordeaux1.fr}
\author{J. Fr{\"o}hlich}
\address[J. Fr{\"o}hlich]{Institut f{\"u}r Theoretische Physik, ETH H{\"o}nggerberg, CH-8093 Z{\"u}rich, Switzerland}
\email{juerg@phys.ethz.ch}
\author{I. M. Sigal}
\address[I.M. Sigal]{Department of Mathematics, University of Toronto, Toronto, ON M5S 2E4, Canada}
\email{im.sigal@utoronto.ca}

\begin{abstract}
We prove the limiting absorption principle for a dressed electron at a fixed total momentum in the standard model of non-relativistic quantum electrodynamics. Our proof is based on an application of the smooth Feshbach-Schur map in conjunction with Mourre's theory.
\end{abstract}

\date{\DATUM}

\maketitle

\section{Introduction}\label{section:intro}

In this paper we study a non-relativistic electron interacting with the quantized electromagnetic field. Since the total momentum of the system is conserved, the Hamiltonian, $H^{SM} := \frac{1}{2} \big ( p - \alpha^{\frac{1}{2}} A(x_{\mathrm{el}}) )^2 + H_f$, can be written as a direct integral, $H^{\mathrm{SM}} = \int_{\mathbb{R}^3}^{\oplus} H(P) \d P$, where the fiber Hamiltonians, $H(P)$, are self-adjoint operators acting on photon Fock space $\mathcal{F}$. (Precise definitions will be given later in this introduction.) We will analyze these fiber Hamiltonians at a fixed total momentum $P \in \mathbb{R}^3$.

We prove the limiting absorption principle (LAP) for $H(P)$, for $|P|$ small enough. As a consequence, we obtain  local decay estimates and absolute continuity of the spectrum of $H(P)$ in the interval $(E(P),+\infty)$, where $E(P)$ denotes the bottom of the spectrum of $H(P)$. The former implies that photons move out of any bounded domain around the electron with probability one, as time tends to infinity.

The quantity $E(P)$ is the energy of a dressed one-particle state of momentum $P$, provided $|P|$ is small enough. Some of its properties have been investigated in several papers (see \cite{Chen, BCFS2, CF, HH, CFP2, FP}). It has been shown that, for $P$ sufficiently small, $H(P)$ has a ground state in the Fock space if and only if $P=0$, (unless an infrared regularization is introduced). This result is sometimes referred to as ``infrared catastrophe''. In \cite{CF}, the existence of a ground state is obtained in a non-Fock coherent infrared representation. The regularity of the map $P \mapsto E(P)$ is studied in \cite{Chen,BCFS2,CFP2,FP}. This yields, among other things, bounds on the renormalized electron mass. We also draw the reader attention to \cite{AFGG}, where related results for a model of a dressed non-relativistic electron in a magnetic field are established.

Our proof of the LAP for $H(P)$ is based on an application of the isospectral \emph{smooth Feshbach-Schur map} introduced in \cite{BCFS} (see also \cite{GH,FGS3}), in conjunction with \emph{Mourre's theory} (see \cite{Mo,PSS,ABG,HS}). For the standard model of charged non-relativistic particles bound to a static nucleus and interacting with the quantized electromagnetic field, a LAP just above the ground state energy has been recently proven in \cite{FGS1,FGS3}. An important ingredient in the proof of \cite{FGS1} is the use of a unitary Pauli-Fierz transformation (combined with exponential decay of states bound to the nucleus in the position variables of the electrons). Such a transformation does not exist in the model considered in the present paper so that the proof of \cite{FGS1} is not applicable.

The method developed in \cite{FGS3} is based on a spectral renormalization group analysis; (see \cite{BFS,BCFS, FGS2}). A similar analysis has been used in \cite{Chen,BCFS2} to study properties of dressed one-particle states in the model studied in the following.
The proof we present in the following is technically simpler, in that we require only one application of the smooth Feshbach-Schur map, whereas renormalization group methods are based on an iteration of this map. The renormalization group might, however, yield somewhat more precise results, in the sense that it is expected to provide an optimal estimate on the Hölder constant appearing in the LAP.

In an appendix, we explain how to modify the proof given in this paper to arrive at a LAP for bound electrons coupled to the radiation field. We emphasize that the infrared singularity in the interaction form factor is not an essential difficulty in our proof.

Next, we describe the model, state our main results and outline the strategy of our proof. Whenever the readers meet an unfamiliar notation they are encouraged to consult Appendix \ref{appendix:notations}. \\

\noindent \textbf{Definition of the model} \\
We consider a freely moving non-relativistic electron interacting with the quantized electromagnetic field. The Hilbert space describing the pure states of the system is given by $\mathcal{H} = \mathcal{H}_{\mathrm{el}} \otimes \mathcal{F}$, where $\mathcal{H}_{\mathrm{el}} = \mathrm{L}^2( \mathbb{R}^3 )$ is the Hilbert space for the electron. For the sake of simplicity, the spin of the electron is neglected. The symmetric Fock space, $\mathcal{F}$, for the photons is defined as
\begin{equation}
\mathcal{F} := \Gamma_s ( \mathrm{L}^2( \mathbb{R}^3 \times \mathbb{Z}_2 ) ) \equiv \mathbb{C} \oplus \bigoplus_{n=1}^\infty S_n \left [ \mathrm{L}^2( \mathbb{R}^3 \times \mathbb{Z}^2 )^{ \otimes^n} \right ],
\end{equation}
where $S_n$ denotes the symmetrization operator on $\mathrm{L}^2( \mathbb{R}^3 \times \mathbb{Z}^2 )^{ \otimes^n}$. As usual, the operators on this space will be expressed in terms of the photon creation and annihilation operators, $a^*_\lambda(k)$, $a_\lambda(k)$, which are operator-valued distributions obeying the canonical commutation relations
\begin{equation}
[ a^{\#}_\lambda(k) , a^{\#}_{\lambda'}(k') ] = 0 , \quad [ a_\lambda(k) , a^*_{\lambda'}(k') ] = \delta_{\lambda\lambda'} \delta( k - k' ),
\end{equation}
where $a^{\#}$ stands for $a^*$ or $a$. As usual, for any $h \in \mathrm{L}^2( \mathbb{R}^3 \times \mathbb{Z}_2 )$, we set
\begin{equation}
a^*(h) := \sum_{\lambda=1,2} \int_{\mathbb{R}^3} h( k,\lambda ) a^*_\lambda(k) \d k, \quad a(h) := \sum_{\lambda=1,2} \int_{\mathbb{R}^3} \bar h( k,\lambda ) a_\lambda( k) \d k.
\end{equation}

In the standard model of non-relativistic QED, the Hamiltonian for a freely moving electron interacting with photons is given by
\begin{equation}
H^{\mathrm{SM}} := \frac{1}{2} \big ( p - \alpha^{\frac{1}{2}} A(x_{\mathrm{el}}) )^2 + H_f,
\end{equation}
acting on $\mathcal{H} = \mathcal{H}_{\mathrm{el}} \otimes \mathcal{F}$. Here $x_{\mathrm{el}}$ denotes the position of the electron and $p=-\i \nabla_{x_{\mathrm{el}}}$ is the electron momentum operator. The Hamiltonian for the free quantized electromagnetic field is given by
\begin{equation}\label{eq:def_Hf}
H_f := \sum_{\lambda=1,2} \int_{\mathbb{R}^3} |k| a^*_\lambda(k) a_\lambda( k) \d k,
\end{equation}
and the electromagnetic vector potential is defined as
\begin{equation}
A(x_{\mathrm{el}}) := \frac{1}{\sqrt{2}} \sum_{\lambda=1,2} \int_{\mathbb{R}^3} \frac{ \kappa^\Lambda(k) }{ |k|^{\frac{1}{2}}} \varepsilon_\lambda(k) ( a^*_\lambda(k) e^{-\i k \cdot x_{\mathrm{el}} } + a_\lambda(k) e^{ \i k \cdot x_{\mathrm{el}} } ) \d k.
\end{equation}
In this expression, the ultraviolet cutoff function $\kappa^\Lambda$ is chosen such that
\begin{equation}\label{eq:kappaLambda}
\kappa^\Lambda \in \mathrm{C}_0^\infty( \{ k , |k| \le \Lambda \} ; [0,1] )\ \mbox{and}\ \kappa^\Lambda =1\ \mbox{on}\ \{ k , |k| \le 3 \Lambda / 4 \}.
\end{equation}
Furthermore, the polarization vectors $\varepsilon_\lambda(k)$, $\lambda=1,2$, are assumed to be real-valued, orthogonal to each other and to $k$.

The system is translation invariant in the sense that $H^{\mathrm{SM}}$ commutes with the total momentum operator $ p + P_f$, where
\begin{equation}
P_f := \sum_{\lambda=1,2} \int_{\mathbb{R}^3} k a^*_\lambda(k) a(k) \d k.
\end{equation}
It follows that $H^{\mathrm{SM}}$ admits the fiber decomposition
\begin{equation}
H^{\mathrm{SM}} = \int_{\mathbb{R}^3}^{\oplus} H(P) \d P,
\end{equation}
where the fiber operators $H(P)$, $P \in \mathbb{R}^3$, are self-adjoint operators on $\mathcal{F}$. The corresponding decomposition of the state space $\mathcal{H} = \mathcal{H}_{\mathrm{el}} \otimes \mathcal{F}$ can be written as $\Psi(x)=\int_{\mathbb{R}^3} e^{\i x\cdot(P-P_f)}\Phi (P) \d P$, with the fibers $\Phi (P)\in \mathcal{F}$. Using that $A(x)e^{\i x\cdot(P-P_f)} = e^{\i x\cdot(P-P_f)}A(0)$, we compute $H \Psi(x)=\int_{\mathbb{R}^3} e^{\i x\cdot(P-P_f)}H(P)\Phi (P) \d P$, where  $H(P)$ are given explicitly by
\begin{equation}
H(P) = \frac{1}{2} \big ( P - P_f - \alpha^{\frac{1}{2}} A \big )^2 + H_f.
\end{equation}
Here
\begin{equation}
A := A(0) = \frac{1}{\sqrt{2}} \sum_{\lambda=1,2} \int_{\mathbb{R}^3} \frac{ \kappa^\Lambda(k) }{ |k|^{\frac{1}{2}}} \varepsilon_\lambda(k) ( a^*_\lambda(k) + a_\lambda(k) ) \d k.
\end{equation}

We define $E(P) := \inf \sigma( H(P) )$. If $\alpha=0$, and if $|P|$ is less than the bare electron mass (equal to 1 in the units used in this paper), then $E(P) = P^2 /2$ is an \emph{eigenvalue} of $H(P)$. If $|P| > 1$, then $E(P) = |P| - 1/2$, and $E(P)$ is \emph{not} an eigenvalue of $H(P)$. The map $P \mapsto E(P)$ is pictured in figure \ref{bottom}, for $\alpha=0$.

\begin{figure}[htbp]
\centering
\resizebox{0.6\textwidth}{!}{\includegraphics{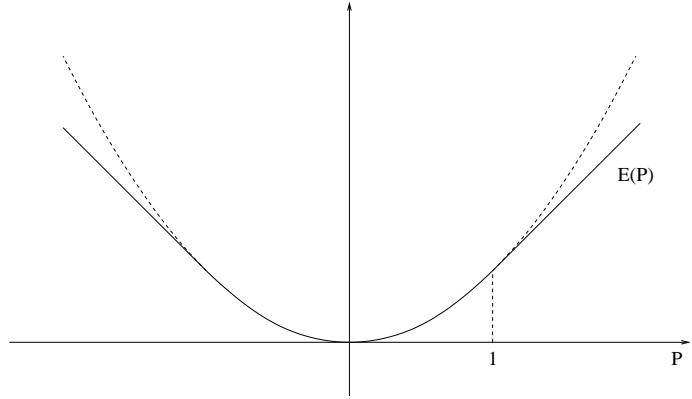}}
\caption{\textbf{The map $E(P) = \inf \sigma ( H(P) )$ for $\alpha=0$}:}
\begin{itemize}
\item[] $\qquad \qquad$ If $|P|\le 1$, $E(P) = P^2/2 \in \sigma_{\mathrm{pp}}( H(P) )$,
\item[] $\qquad \qquad$ If $|P|>1$, $E(P) = |P| - 1/2 \notin \sigma_{\mathrm{pp}}( H(P) )$.
\end{itemize}
 \label{bottom}
\end{figure}

We prove a limiting absorption principle for $H(P)$ in an energy interval just above $E(P)$, for $|P| \le p_c$, where $0 < p_c < 1$. In this paper we choose $p_c = 1/40$, and we do not attempt to find an optimal value for $p_c$.
\\

\noindent \textbf{Main results} \\
For an interval $J$, we set $J_{\pm} = \big \{ z \in \mathbb{C} , \mathrm{Re}z \in J , 0 < \pm \mathrm{Im}z \le 1 \big \}$. Let
\begin{equation}
y := \i \nabla_k
\end{equation}
be the observable accounting for the ``position''
of the photon relative to the electron position. Let $\d \Gamma( b )$ denote the usual (Lie-algebra) second quantization of an operator $b$ acting on $\mathrm{L}^2( \mathbb{R}^3 \times \mathbb{Z}_2 )$.
Our goal is to prove the following result.
\begin{theorem}\label{thm:LAP}
There exists an $\alpha_0>0$ such that, for any $|P| \le p_c \, ( \, \le 1/40)$, $0 \le \alpha \le \alpha_0$, $1/2 < s \le 1$, and any compact interval $J \subset (E(P),\infty)$, we have that
\begin{equation}\label{eq:main_LAP}
\sup_{ z \in J_\pm } \big \| (  \d \Gamma ( \langle y \rangle ) + 1 )^{- s} \big [ H(P) - z \big ]^{-1} ( \d \Gamma ( \langle y \rangle )  + 1 )^{-s} \big \| \le \mathrm{C},
\end{equation}
where $\mathrm{C}$ is a constant depending on $J$ and $s$.
Moreover, the map
\begin{equation}
J \ni \lambda \mapsto ( \d \Gamma ( \langle y \rangle ) + 1 )^{- s} \big [ H(P) - \lambda \pm \i 0 \big ]^{-1} ( \d \Gamma ( \langle y \rangle ) + 1 )^{- s} \in B( \mathcal{H} )
\end{equation}
is uniformly H\"older continuous in $\lambda$ of order $s-1/2$.
\end{theorem}

This theorem follows from Corollaries \ref{cor:LAPaway_2} and \ref{cor:LAP_K_2} below. Our proof will show that, if $\mathrm{dist} ( E(P) , J ) = \sigma$ then the constant $\mathrm{C}$ in \eqref{eq:main_LAP} is bounded by $ O( \sigma^{-1} )$. Finding an optimal upper bound on $\mathrm{C}$ with respect to $\sigma$ is beyond the scope of this paper.

As a consequence of Theorem \ref{thm:LAP}, we obtain the following
\begin{corollary}\label{cor:absCont}
There exists $\alpha_0>0$ such that for any $|P| \le p_c$ and $0 \le \alpha \le \alpha_0$, the spectrum of $H(P)$ is purely absolutely continuous in the interval $(E(P) , + \infty )$.
\end{corollary}
\,

\noindent\textbf{Physical interpretation of the results}\\
Next, we describe a consequence of  Theorem \ref{thm:LAP}
related to a key physical property of the system.
We consider an initial state consisting of a dressed electron
together with a cloud of additional photons supported in a finite ball centered at the position of the electron.
As we demonstrate below, Theorem \ref{thm:LAP} implies that
asymptotically, as  $t\rightarrow\infty$, all photons disperse to spatial infinity. 

\begin{corollary}
Let  ${\mathcal{S}}:=\{P\in\mathbb{R}^3| \, |P|<p_c\}$,
and let $\Phi=\int_{\mathcal{S}}^\oplus \d P \, \Phi(P)
\in\mathcal{H} = \mathcal{H}_{\mathrm{el}} \otimes \mathcal{F}$ denote an arbitrary state,
satisfying
\begin{equation}
	\| \, (\d\Gamma(\langle y\rangle)+1)^s \, \Phi(P) \, \| \, < \, \infty \,,
\end{equation}
for some $1/2<s \le 1$ and for all  $P\in\mathcal{S}$.
Furthermore, let $f\in C^\infty_0(\mathcal{M}_{a.c.})$
be a smooth, compactly supported function on the set
\begin{equation}
	\mathcal{M}_{a.c.}\, := \{(\lambda,P) \in \mathbb{R}\times \mathcal{S}\, | \, \lambda>E(P)\} \,.
\end{equation}
Finally, let $x_{\mathrm{ph}}$ denote the photon ``position'' operator, $x_{\mathrm{ph}} := x_{\mathrm{el}}+y$. Then
\begin{eqnarray} \label{LD}
	\| \, (\d\Gamma(\langle x_{\mathrm{ph}} - x_{\mathrm{el}} \rangle)+1)^{-s} \, e^{-\i tH} \, f(H,P)\Phi \, \|
	\, \leq \, \mathrm{C}  \, t^{-(s-\frac12)} \,.
\end{eqnarray}
\end{corollary}
\begin{proof} Let	$\Phi_f \, := \,  f(H,P) \,  \Phi \,.$ 
The state $\Phi_f\in\mathcal{H}$ can be written as
\begin{equation}
	\Phi_f \, = \, f(H,P)\Phi \, = \, \int_{\mathcal{S}}^\oplus \d P \, f(H(P),P) \, \Phi(P) \,.
\end{equation}
We recall that, in Theorem \ref{thm:LAP}, the observable $y$ is the relative position
of a photon with respect to the electron. We note that
\begin{eqnarray}
	e^{-\i t H}\Phi_{f}
	\, = \, \lim_{\varepsilon\rightarrow0} \frac{1}{2\i\pi}
	\int_{\mathcal{S}} \d P \int \d\lambda \, f(\lambda,P) \, e^{-\i t\lambda}
	{\rm Im}\frac{1}{H(P)-\lambda-\i \varepsilon}   \Phi(P) \, ,
\end{eqnarray}
so that
\begin{align}
&	\| \, (\d\Gamma(\langle x_{\mathrm{ph}} - x_{\mathrm{el}} \rangle)+1)^{-s} \, e^{-\i tH} \, \Phi_{f} \, \| \notag \\
	& =   \, \sup_{\|\Phi'\|=1}\Big| \,
	\lim_{\varepsilon\rightarrow0} \int_{\mathcal{S}} \d P \int \d\lambda \, e^{-\i t\lambda} \,
	f(\lambda,P) \notag\\
	&\qquad	\Big\langle \, \Phi' \, , \, (\d\Gamma(\langle y\rangle)+1)^{-s}
	{\rm Im}\frac{1}{H(P)-\lambda-\i\varepsilon} \Phi(P) 
 \, \Big\rangle \, \Big| \,.
	\label{eq-decaynorm-1}
\end{align}
 Since $f(\lambda,P)$ is supported on the set $\{\lambda>E(P)\}$, Theorem \ref{thm:LAP} implies that the expectation $\langle \cdots\rangle$ in \eqref{eq-decaynorm-1}
is $(s-\frac12)$-H\"older continuous in $\lambda$, for any choice of $\Phi'$, and
for a H\"older constant independent of $\Phi'$, because $\widetilde\Phi(P):=(\d\Gamma(\langle y\rangle)+1)^{ s}\Phi(P) \in \mathcal{F}$.
The Fourier transform $\widehat{g}(t)=\int \d\lambda e^{it\lambda}g(\lambda)$
of an $(s-\frac12)$-H\"older continuous function $g(\lambda)$ satisfies
$|\widehat g(t)|\leq \mathrm{C} t^{-(s-1/2)}$.
 Thus, \eqref{LD} follows.
\end{proof}

As was mentioned above, this corollary implies that photons which are not permanently bound to the electron move out of any bounded domain around the electron with probability one, as time tends to $\infty$.

We consider an observable $A$, given by
a selfadjoint operator on $\mathcal{H}$
which  we assume to satisfy
\begin{equation}
	\|(\d\Gamma(\langle x_{\mathrm{ph}} - x_{\mathrm{el}} \rangle)+1)^s A (\d\Gamma(\langle x_{\mathrm{ph}} - x_{\mathrm{el}} \rangle)+1)^s\| \, < \, \infty \,,
\end{equation}
Then,
\begin{equation}\label{eq-asymprelax-1}
	\lim_{t\rightarrow0}\Big\langle \,  \Phi_f \, , \, e^{ \i tH} \, A
	\, e^{ -\i tH} \, \Phi_f \, \, \Big\rangle
	\, = \, 0 \,.
\end{equation}
Indeed,  we have
\DETAILS{\begin{eqnarray}
	\mathcal{R}(t) \, := \,
	\Big\langle \, e^{-\i tH}\Phi_f \,,\, A \, e^{-\i tH}\Phi_f \Big\rangle \,.
\end{eqnarray}
To bound $\mathcal{R}(t)$}
\begin{align}
	|\Big\langle \, e^{-\i tH}\Phi_f \,,\, A \, e^{-\i tH}\Phi_f \Big\rangle| &\le \|(\d\Gamma(\langle x_{\mathrm{ph}} - x_{\mathrm{el}} \rangle)+1)^{s} A (\d\Gamma(\langle x_{\mathrm{ph}} - x_{\mathrm{el}} \rangle)+1)^{s}\| \notag \\
	&\quad 	\times \| \, (\d\Gamma(\langle x_{\mathrm{ph}} - x_{\mathrm{el}} \rangle)+1)^{-s} \, e^{-\i tH} \, \Phi_{f} \, \|^2 \notag \\
	&\leq \mathrm{C}  \, t^{-2(s-\frac12)}.
\end{align}

More generally, we expect the following picture to hold true.
We assume that $h \in C^\infty( ( - \infty , E_c ) \times \mathcal{S})$, where $E_c=E(P)$ with  $|P|=p_c$, and consider
the state
$$
	\Phi_h \, := \, h(H,P) \, \Phi
$$
where $\Phi\in\mathcal{H}$ is as above.
Let $A=\int^\oplus \d P \, A_P$ denote a translation invariant observable.
Then, we expect that
\begin{equation}
	\lim_{t\rightarrow\infty}\Big\langle \, e^{\i tH}\Phi_h \,,\, A \, e^{\i tH}\Phi_h \, \Big\rangle
	\, = \,  \int_{\mathcal{S}} \d\mu_{\Phi_h}( P) \,
	\Big\langle \Psi_P \, , \, A_P \, \Psi_P \, \Big\rangle \,,
\end{equation}
where ${\rm supp}\{\d \mu_{\Phi_h}\}\subseteq \mathcal{S}$.
Here, $\langle \, \Psi_P \, ,  ( \, \cdot \, ) \, \Psi_P \, \rangle$ denotes an
expectation in the generalized ground state of the fiber Hamiltonian $H(P)$.
This describes the relaxation of the state $\Phi_h$ to the mass shell,
asymptotically as $t\rightarrow\infty$,
under emission of photons which disperse to spatial infinity. (Note that, for $P \neq 0$, $\Psi_P$ does not belong to the Fock space, but to a Hilbert space carrying an infrared representation of the canonical commutation relations.)
\\

\noindent \textbf{Strategy of the proof} \\
The proof of Theorem \ref{thm:LAP} is divided into two steps. First, we prove a LAP in any compact interval $J \subset (E(P),\infty)$ with the property that $\inf J \ge E(P) + \mathrm{C}_0 \alpha^{1/2}$, where $\mathrm{C}_0$ is a fixed, sufficiently large, positive constant.  This follows from a Mourre estimate of the form
\begin{equation}\label{eq:strategy_Mourre1}
\mathds{1}_J(H(P)) [ H(P) , \i B ] \mathds{1}_J( H(P) ) \ge \mathrm{c} \mathds{1}_J( H(P) ),
\end{equation}
where $B$ is the generator of dilatations on Fock space (see Equation \eqref{B}) and $\mathrm{c}$ is positive. Using the assumption that $\inf J \ge E(P) + \mathrm{C}_0 \alpha^{1/2}$ and standard estimates, Equation \eqref{eq:strategy_Mourre1} can be proven in a straightforward way.

In a second, more difficult step, we prove a limiting absorption principle near $E(P)$. We use a theorem due to \cite{FGS3} (see Theorem \ref{thm:LAPtransfer} in Appendix \ref{appendix:Feshbach} of the present paper), which essentially says that one can derive a LAP for $H(P)$ from a LAP for an operator resulting from applying a smooth Feshbach-Schur map to $H(P)$.

Our construction of the smooth Feshbach-Schur map is based on a \emph{low-energy decomposition} of the Hamiltonian $H(P)$:
\begin{equation}
H(P)=H_\sigma(P) +U_\sigma(P),
\end{equation}
where  $\sigma \ge 0$,  $U_\sigma(P)$ is defined by this equation and the infrared cutoff Hamiltonian $H_\sigma(P)$, $\sigma \ge 0$, is given by
\begin{equation}\label{eq-HsigmaP-def-0}
H_\sigma(P) := \frac{1}{2} ( P - P_f - \alpha^{\frac{1}{2}} A_\sigma )^2 + H_f,
\end{equation}
for every $P \in \mathbb{R}^3$, with
\begin{equation}
A_\sigma :=  \frac{1}{\sqrt{2}} \sum_{\lambda=1,2} \int_{ \{ |k| \ge \sigma \} } \frac{ \kappa^\Lambda(k) }{ |k|^{\frac{1}{2}}} \varepsilon_\lambda(k) ( a^*_\lambda(k) + a_\lambda(k) ) \d k,
\end{equation}
(see Section \ref{section:decomp}).
Note that $H_{\sigma}(P)$ is cutoff in the infrared, with the property that photons of energy less than $\sigma$ do not interact with the electron.
Such a decomposition was used previously in the analysis of non-relativistic QED; (see, e.g., \cite{BFP,FGS1}).

Next we use the fact that the Hilbert space $\mathcal{F}$ is unitarily equivalent to $\mathcal{F}_\sigma \otimes \mathcal{F}^\sigma$ where $\mathcal{F}_\sigma := \Gamma_s( \mathrm{L}^2( \{ (k,\lambda) , |k| \ge \sigma \} ) )$ and $\mathcal{F}^\sigma := \Gamma_s( \mathrm{L}^2 ( \{ (k,\lambda)  , |k| \le \sigma \} ) )$. Below we will use this representation without always mentioning it.
The operator $H_\sigma(P)$ leaves invariant the Fock space $\mathcal{F}_\sigma $ of photons of energies larger than $\sigma$, and its restriction to $\mathcal{F}_\sigma$,
\begin{equation}
K_\sigma(P) := H_\sigma(P) |_{\mathcal{F}_\sigma},
\end{equation}
  has a gap of order $O(\sigma)$ in its spectrum above the ground state energy. We use the projection, $P_\sigma(P)$, onto the ground state of $K_\sigma(P)$ in order to construct the smooth Feshbach-Schur map $F_\chi$, where $\chi = P_\sigma(P) \otimes \chi^\sigma_f(H_f)$, with $\chi^\sigma_f(H_f)$ a smooth ``projection'' onto the spectral subspace $H_f \le \sigma$; (see Section \ref{section:Feshbach}). This map projects out the degrees of freedom corresponding to photons of energies larger than $\sigma$. The resulting operator $F (\lambda):=F_\chi(H(P)-\lambda)$, where $\lambda$ is the spectral parameter, is of the form
  \begin{align}
F(\lambda) =& K_\sigma(P) \otimes\mathds{1} + \mathds{1} \otimes \big( \frac{1}{2} P_f^2 + H_f \big ) -  \nabla E_\sigma(P) \otimes P_f  - \lambda +W,
\end{align}
where the operator $W$ can be estimated by  $O(\alpha^{1/2}\sigma)$, and $E_\sigma(P) := \inf \sigma( H_\sigma(P) )$.

Next, in order to obtain a LAP for $F(\lambda)$, we use again Mourre's theory, choosing a conjugate operator $B^\sigma$ defined as the generator of dilatations with a cutoff in the photon momentum variable:
\begin{align}\label{eq:Bsigma'}
& B^\sigma :=  \sum_{\lambda=1,2} \int_{\mathbb{R}^3}  a^*_\lambda(k) \kappa^\sigma b \kappa^\sigma a_\lambda( k) \d k,
\end{align}
with $\kappa^\sigma$ a cutoff in the photon momentum variable defined in \eqref{eq:kappaLambda}, and $ b := \frac{ \i }{2} ( k \cdot \nabla_k + \nabla_k \cdot k )$, the generator of dilatations (see Section \ref{section:mourre}). Let $\lambda$ be in the interval
\begin{equation}
J_\sigma^< := [ E(P) + 11 \rho \sigma / 128 , E(P) + 13 \rho \sigma / 128 ],
\end{equation}
where $\sigma$ satisfies $\sigma \le \mathrm{C}'_0 \alpha^{1/2}$ for some fixed, sufficiently large positive constant $\mathrm{C}'_0 \ge \mathrm{C}_0$, and $\rho \sigma$ is the size of the gap above $E_\sigma(P)$ in the spectrum of $K_\sigma(P)$. The Mourre estimate for $F(\lambda)$, on the spectral interval $\Delta_\sigma = [ - \rho \sigma / 128 , \rho \sigma / 128 ]$, is established as follows. By energy localization and the facts that the operator $K_\sigma(P)$ commutes with $B^\sigma$ and that $| \nabla E_\sigma(P) | \le |P| + \mathrm{C} \alpha \le 1/4$ for $|P| \le 1/40$ and $\alpha$ sufficiently small, the commutator of the unperturbed part in $F(\lambda)$ with $B^\sigma$ gives  a positive term of order $O(\sigma)$. This and the fact that the commutator with the perturbation $W$ is of order $O(\alpha^{1/2}\sigma)$ yield the Mourre estimate, and therefore the LAP, for $F(\lambda)$.  Once the LAP is established for $F(\lambda)$, it is transferred by the theorem of \cite{FGS3}, mentioned above, to the original Hamiltonian $H(P)$ on the interval $ J_\sigma^<$.  Finally, we use that the intervals above with $\sigma \le \mathrm{C}'_0 \alpha^{1/2}$ cover the interval $(E(P), \mathrm{C}_0\alpha^{1/2}]$.   \\

\noindent \textbf{Organization of the paper} \\
Our paper is organized as follows. In the next section, we prove the LAP for $H(P)$ outside a certain neighborhood of the ground state energy. Section \ref{section:decomp} is concerned with the approximation of $H(P)$ by the infrared cutoff Hamiltonian $H_\sigma(P)$. In Section \ref{section:Feshbach}, we prove the existence of the Feshbach-Schur operator $F(\lambda)$ mentioned above. We establish the Mourre estimate for $F(\lambda)$ in Section \ref{section:mourre}, from which we deduce the LAP for $H(P)$ near $E(P)$. In Appendix \ref{appendix:estimates}, we collect some technical estimates used in Sections \ref{section:Feshbach} and \ref{section:mourre}. Appendix \ref{appendix:Feshbach} recalls the definition of the smooth Feshbach-Schur map and some of its main properties. In Appendix \ref{appendix:bound}, we briefly explain how to adapt the method used in this paper to a model of bound non-relativistic electrons coupled to the radiation field. Finally, for the convenience of the reader, a list of notations used in this paper is contained in Appendix \ref{appendix:notations}.

Throughout the paper, $\mathrm{C}, \mathrm{C}', \mathrm{C}''$ will denote positive constants that may differ from one line to another.
\\

\noindent \textbf{Acknowledgements} \\
J.Fa., I.M.S., and T.C. are grateful to J.Fr. for hospitality at ETH Z{\"u}rich. T.C. thanks I.M.S. for his hospitality at the University of Toronto. The authors acknowledge the support of the Oberwolfach Institute. A part of this work was done during I.M.S.'s stay at the IAS, Princeton. Research of I.M.S. is supported by NSERC under Grant NA 7901. T.C. was supported by NSF grant DMS-070403/DMS-0940145.

\section{Limiting absorption principle outside a neighborhood of the ground state energy}\label{section:LAP_away}

In this section we shall prove Theorem \ref{thm:LAP} for any interval $J$ of the form $$J = J_\sigma^> := E(P) + [ \sigma , 2 \sigma ],$$ where the parameter $\sigma$ is chosen to satisfy $\sigma \ge \mathrm{C}_0 \alpha^{\frac{1}{2}}$, for some fixed positive constant $\mathrm{C}_0$. Our proof is based on the standard Mourre theory (\cite{Mo}), the conjugate operator $B$ being chosen as the generator of dilatations  on $\mathcal{F}$, i.e.,
\begin{equation} \label{B}
B := \d \Gamma ( b ), \quad\text{with}\quad b := \frac{ \i }{2} ( k \cdot \nabla_k + \nabla_k \cdot k ).
\end{equation}
One can easily check that
\begin{equation}\label{eq:[Hf,iB]_away}
[ H_f , \i B ] = H_f, \quad [ P_f , \i B ] = P_f,
\end{equation}
and, for any $f \in D(b)$,
\begin{equation}\label{eq:[Phi,iB]_away}
[ \Phi(f) , \i B ] = - \Phi ( \i b f ),
\end{equation}
where
\begin{equation}
\Phi(h) := \frac{1}{\sqrt{2}} ( a^*(h) + a(h) ),
\end{equation}
so that
\begin{equation}
A = \Phi ( h ), \quad h( k,\lambda ) := \frac{ \kappa^\Lambda(k) }{ |k|^{\frac{1}{2} } } \varepsilon_\lambda(k).
\end{equation}
\begin{theorem}\label{thm:mourre_away}
There exist constants $\alpha_0 > 0$ and $\mathrm{C}_0 > 0$ such that, for all $|P| \le p_c$, $0 \le \alpha \le \alpha_0$ and $\sigma \ge \mathrm{C}_0 \alpha^{1/2}$,
\begin{equation}
\mathds{1}_{J_\sigma^>} ( H(P) ) [ H(P) , \i B ] \mathds{1}_{J_\sigma^> } ( H(P) ) \ge \frac{\sigma}{2} \mathds{1}_{J_\sigma^> } ( H(P) ).
\end{equation}
\end{theorem}
\begin{proof}
Note that $H(P)$ can be written as
\begin{align}
H(P) =& \frac{1}{2} P^2 + \frac{1}{2} P_f^2 + H_f - P \cdot P_f - \alpha^{\frac{1}{2}} P \cdot \Phi (h)  \notag \\
&+ \frac{ \alpha^{\frac{1}{2}} }{2} \big ( \Phi( h ) \cdot P_f  + P_f \cdot \Phi(h) \big ) + \frac{ \alpha }{2} \Phi(h)^2. \label{eq:H(P)_1}
\end{align}
It follows from \eqref{eq:[Hf,iB]_away} and \eqref{eq:[Phi,iB]_away} that
\begin{align}
[ H(P) , \i B ]  =& - \frac{1}{2} \big ( P - P_f - \alpha^{\frac{1}{2}} \Phi(h) \big ) \cdot \big ( P_f - \alpha^{\frac{1}{2}} \Phi( \i b h ) \big ) \notag \\
&- \frac{1}{2} \big ( P_f - \alpha^{\frac{1}{2}} \Phi( \i b h ) \big ) \cdot \big ( P - P_f - \alpha^{\frac{1}{2}} \Phi(h) \big ) + H_f. \label{eq:[H(P),iB]}
\end{align}
By \eqref{eq:H(P)_1} we get
\begin{align}
[ H(P) , \i B ] \ge& H(P) - \frac{1}{2} P^2 + \alpha^{\frac{1}{2}} P \cdot \big ( \Phi( h ) + \Phi ( \i b h ) \big ) - \frac{ \alpha }{2} \Phi(h)^2 \notag \\
& - \frac{ \alpha^{\frac{1}{2}} }{2} \big ( \Phi( \i b h ) \cdot ( P_f + \alpha^{\frac{1}{2}} \Phi(h) ) + ( P_f + \alpha^{\frac{1}{2}} \Phi(h) ) \cdot \Phi ( \i b h ) \big ). \label{eq:mourre_away3}
\end{align}
Multiplying both sides of Inequality \eqref{eq:mourre_away3} by $\mathds{1}_{J_\sigma^>}(H(P))$, using in particular that $\Phi(h)$ and $\Phi( \i b h )$ are $H(P)$-bounded, this yields
\begin{equation}
\mathds{1}_{J_\sigma^>} ( H(P) ) [ H(P) , \i B ] \mathds{1}_{J_\sigma^> } ( H(P) ) \ge \big ( E(P) - \frac{1}{2} P^2 + \sigma - \mathrm{C} \alpha^{ \frac{1}{2} } \big ) \mathds{1}_{J_\sigma^> } ( H(P) ).
\end{equation}
Since, by Proposition \ref{prop:Esigma(P)}, $ | E(P) - P^2/2 | \le \mathrm{C}' \alpha $, we obtain
\begin{align}
\mathds{1}_{J_\sigma^>} ( H(P) ) [ H(P) , \i B ] \mathds{1}_{J_\sigma^> } ( H(P) ) & \ge \big ( \sigma - \mathrm{C}'' \alpha^{ \frac{1}{2} } \big ) \mathds{1}_{J_\sigma^>} ( H(P) ) \notag \\
& \ge \frac{ \sigma }{2} \mathds{1}_{J_\sigma^>} ( H(P) ),
\end{align}
provided that $ \sigma \ge \mathrm{C}_0 \alpha^{1/2}$, the constant $\mathrm{C}_0$ being chosen sufficiently large.
\end{proof}
\begin{corollary} \label{cor:LAPaway}
There exists $\alpha_0>0$ such that, for any $|P| \le p_c$, $0 \le \alpha \le \alpha_0$ and $1/2 < s \le 1$, and for any compact interval $J \subset [ E(P) + \mathrm{C}_0 \alpha^{1/2} , \infty )$,
\begin{equation}
\sup_{z \in J_{\pm}} \big \| \langle B \rangle^{- s} \big [ H(P) - z \big ]^{-1} \langle B \rangle^{-s} \big \| < \infty.
\end{equation}
Here $\mathrm{C}_0>0$ is given by Theorem \ref{thm:mourre_away}. Moreover, the map
\begin{equation}
J \ni \lambda \mapsto \langle B \rangle^{- s} \big [ H(P) -  \lambda \pm \i 0 ]^{-1} \langle B \rangle^{-s} \in B( \mathcal{H} )
\end{equation}
is uniformly H\"older continuous in $\lambda$ of order $s-1/2$.
\end{corollary}
\begin{proof}
Using the well-known conjugate operator method (see \cite{Mo}, \cite{ABG}), it suffices to show that $H(P) \in \mathrm{C}^2(B)$. Since $D(H(P)) = D(P_f^2/2+H_f)$, one can check, in the same way as in \cite[Proposition 9]{FGS1}, that for all $t \in \mathbb{R}$,
\begin{equation}
e^{\i t B } D( H(P) ) \subset D( H(P) ).
\end{equation}
Therefore, in order to obtain the $\mathrm{C}^2$-property of $H(P)$ with respect to $B$, it is sufficient to verify that $[ H(P) , \i B ]$ and $[ [ H(P) , \i B ] , \i B ]$ extend to $H(P)$-bounded operators. This follows easily from the expression of the commutator of $H(P)$ with $\i B$, Equation \eqref{eq:[H(P),iB]}, and by computing similarly the double commutator $[ [ H(P) , \i B ], \i B]$.
\end{proof}
\begin{corollary} \label{cor:LAPaway_2}
Under the conditions of Corollary \ref{cor:LAPaway},
\begin{equation}\label{cor:LAPaway_2_eq1}
\sup_{z \in J_{\pm}} \big \| ( \d \Gamma (\langle y \rangle) + 1)^{- s} \big [ H(P) - z \big ]^{-1} ( \d \Gamma (\langle y \rangle) + 1)^{-s} \big \| < \infty,
\end{equation}
and the map
\begin{equation}\label{cor:LAPaway_2_eq2}
J \ni \lambda \mapsto ( \d \Gamma (\langle y \rangle) + 1)^{- s} \big [ H(P) - \lambda \pm \i 0 \big ]^{-1} ( \d \Gamma (\langle y \rangle) + 1)^{-s} \in B( \mathcal{H} )
\end{equation}
is uniformly H\"older continuous in $\lambda$ of order $s-1/2$.
\end{corollary}
\begin{proof}
Let $\phi \in \mathrm{C}_0^\infty( \mathbb{R} ; [0,1] )$ be such that $\phi=1$ on $[ E(P) , \sup J ]$ and $\mathrm{supp}(\phi) \subset (-\infty , \frac{3}{2} \sup J)$. Let $\bar \phi = 1 - \phi$. We have
\begin{equation}
\sup_{z \in J_{\pm}} \big \| \bar \phi ( H(P) ) \big [ H(P) - z \big ]^{-1} \big \| < \infty.
\end{equation}
Therefore, to establish \eqref{cor:LAPaway_2_eq1}, it suffices to prove that
\begin{equation}\label{eq:LAP_2}
\sup_{z \in J_{\pm}} \big \| ( \d \Gamma (\langle y \rangle) + 1)^{- s} \phi( H(P) ) \big [ H(P) - z \big ]^{-1} ( \d \Gamma (\langle y \rangle) + 1)^{-s} \big \| < \infty.
\end{equation}

Recall that $\kappa^\Lambda$ denotes a function in $\mathrm{C}_0^\infty( \{ k , |k| \le \Lambda \} ; [0,1] )$ chosen such that $\kappa^\Lambda =1$ on $\{ k , |k| \le 3 \Lambda / 4 \}$. Let $\tilde \Lambda := \max( \Lambda , 2 \sup J )$ and let $\mathcal{U}$ denote the  unitary operator identifying $\mathcal{F}$ and $\mathcal{F}_{\tilde \Lambda} \otimes \mathcal{F}^{\tilde \Lambda}$,
\begin{equation}
\mathcal{U} : \mathcal{F} \to \mathcal{F}_{\tilde \Lambda} \otimes \mathcal{F}^{\tilde \Lambda},
\end{equation}
where $\mathcal{F}_{\tilde \Lambda} := \Gamma_s( \mathrm{L}^2 ( \{ ( k , \lambda ) , |k| \ge \tilde \Lambda \} ) )$ and $\mathcal{F}^{\tilde \Lambda} := \Gamma_s( \mathrm{L}^2 ( \{ ( k , \lambda ) , |k| \le \tilde \Lambda \} ) )$. Let $P_\Omega$ denote the projection onto the Fock vacuum and $\bar P_\Omega := \mathds{1} - P_\Omega$. Since $A = \mathcal{U}^* ( \mathds{1} \otimes A ) \mathcal{U}$, it follows that $H(P)$ commutes with $\mathcal{U}^* ( P_\Omega \otimes \mathds{1} ) \mathcal{U}$. Moreover, we have that
\begin{align}
H(P) \, \mathcal{U}^* ( \bar P_\Omega \otimes \mathds{1} ) \mathcal{U} & \ge H_f \, \mathcal{U}^* ( \bar P_\Omega \otimes \mathds{1} ) \mathcal{U} \notag \\
&= \mathcal{U}^* ( ( H_f \bar P_\Omega ) \otimes \mathds{1} + \mathds{1} \otimes H_f ) \mathcal{U} \notag \\
& \ge \mathcal{U}^* ( ( H_f \bar P_\Omega ) \otimes \mathds{1} ) \mathcal{U} \notag \\
& \ge \tilde \Lambda \, \mathcal{U}^* ( \bar P_\Omega \otimes \mathds{1} ) \mathcal{U},
\end{align}
and hence, since $\mathrm{supp} (\phi) \subset (-\infty , \frac{3}{2} \sup J)$ with $\frac{3}{2} \sup J < \tilde \Lambda$, we obtain that
\begin{equation}\label{eq:phi(H(P))Gamma_1}
\phi( H(P) ) = \phi( H(P) ) \, \mathcal{U}^* ( P_\Omega \otimes \mathds{1} ) \mathcal{U}.
\end{equation}
Recall that, given an operator $a$ on $\mathrm{L}^2( \mathbb{R}^3 \times \mathbb{Z}_2 )$, the operator $\Gamma( a )$ on $\mathcal{F}$ is defined by its restriction to the $n$-particle sector as
\begin{equation}
\Gamma( a ) |_{ S_n \mathrm{L}^2( \mathbb{R}^3 \times \mathbb{Z}_2 )^{\otimes n} } = a \otimes \cdots \otimes a.
\end{equation}
From \eqref{eq:phi(H(P))Gamma_1} and the fact that $\mathcal{U}^* ( P_\Omega \otimes \mathds{1} ) \mathcal{U} = \mathcal{U}^* ( P_\Omega \otimes \mathds{1} ) \mathcal{U} \Gamma( \kappa^{2 \tilde \Lambda} )$, we obtain
\begin{equation}\label{eq:phi(H(P))Gamma_2}
\phi( H(P) ) = \phi ( H(P) ) \Gamma ( \kappa^{2 \tilde \Lambda } ).
\end{equation}
Note that the advantage in using \eqref{eq:phi(H(P))Gamma_2} rather than \eqref{eq:phi(H(P))Gamma_1} in what follows comes from the fact that the function $\kappa^{2 \tilde \Lambda}$ is smooth.

Considering the restriction of the operator below to all $n$-particles subspaces of the Fock space, one verifies that
\begin{equation}
\big \| B \, \Gamma( \kappa^{2 \tilde \Lambda } ) ( \d \Gamma (\langle y \rangle) + 1)^{- 1} \big \| \le \mathrm{C}.
\end{equation}
Using an interpolation argument, this implies
\begin{equation}\label{eq:B_dGamma}
\big \| \langle B \rangle^s \, \Gamma( \kappa^{2 \tilde \Lambda} ) ( \d \Gamma (\langle y \rangle) + 1)^{- s} \big \| \le \mathrm{C},
\end{equation}
for any $0 \le s \le 1$.
Combining Corollary \ref{cor:LAPaway} with \eqref{eq:phi(H(P))Gamma_2} and \eqref{eq:B_dGamma}, we obtain \eqref{eq:LAP_2}, which concludes the proof of \eqref{cor:LAPaway_2_eq1}. The H{\"o}lder continuity stated in \eqref{cor:LAPaway_2_eq2} follows similarly. 
\end{proof}
Henceforth and throughout the remainder of this paper, we assume that
\begin{equation}
\sigma \le \mathrm{C}'_0 \alpha^{\frac{1}{2}},
\end{equation}
where $\mathrm{C}'_0$ is a positive constant such that $\mathrm{C}'_0 \ge \mathrm{C}_0$ (here $\mathrm{C}_0$ is given by Theorem \ref{thm:mourre_away}).

\section{Low energy decomposition}\label{section:decomp}

For $\sigma \ge 0$ we define the infrared cutoff Hamiltonian $H_\sigma(P)$ by
\begin{equation}\label{eq-HsigmaP-def-1}
H_\sigma(P) := \frac{1}{2} ( P - P_f - \alpha^{\frac{1}{2}} A_\sigma )^2 + H_f,
\end{equation}
where
\begin{equation}
A_\sigma := \Phi( h_\sigma ), \quad h_\sigma( k,\lambda ) := \frac{ \kappa_\sigma^\Lambda(k) }{ |k|^{\frac{1}{2} } } \varepsilon_\lambda(k),
\end{equation}
and
\begin{equation}
\kappa_\sigma^\Lambda(k) := \mathds{1}_{ \{ |k| \ge \sigma \} } (k) \kappa^\Lambda(k).
\end{equation}
Note that $H_0(P) = H(P)$. Let
\begin{equation}
E_\sigma(P) := \inf \sigma( H_\sigma(P) ).
\end{equation}
For $\sigma = 0$ we set $E(P) := E_0(P)$. Let $\mathcal{F}_\sigma := \Gamma_s( \mathrm{L}^2( \{ (k,\lambda) , |k| \ge \sigma \} ) )$ and
\begin{equation}
K_\sigma(P) := H_\sigma(P) |_{\mathcal{F}_\sigma}.
\end{equation}

Let $\mathrm{Gap}( H )$ be defined by $\mathrm{Gap}( H ) := \inf \{ \sigma(H) \setminus \{ E(H) \} \} - E(H)$, where $E(H) := \inf \{ \sigma(H)  \}$, for any self-adjoint and semi-bounded operator $H$. The following proposition is proven in \cite{Chen,BCFS2,CFP2,FP}:
\begin{proposition}\label{prop:Esigma(P)}
There exists $\alpha_0>0$ such that for all $0\le \alpha \le \alpha_0$, the following properties hold:
\begin{itemize}
\item[1)] For all $\sigma > 0$ and $|P| \le p_c$,
\begin{equation} \label{gapineq}
\mathrm{Gap}( K_\sigma(P) ) \ge \rho \sigma \text{ for some } 0 < \rho < 1.
\end{equation}
Moreover $\inf \sigma ( K_\sigma(P) ) = E_\sigma(P)$ is a non-degenerate (isolated) eigenvalue of $K_\sigma(P)$.
\item[2)] For all $\sigma \ge 0$ and $|P| \le p_c$,
\begin{equation}
\big | E_\sigma(P) - E(P) \big | \le \mathrm{C} \alpha \sigma,
\end{equation}
where $\mathrm{C}$ is a positive constant independent of $\sigma$.
\item[3)] For all $\sigma > 0$, the map $P \mapsto E_\sigma(P)$ is twice continuously differentiable on $\{ P \in \mathbb{R}^3 , |P| \le p_c \}$ and satisfies
\begin{align}
& \big | E_\sigma(P) - \frac{P^2}{2} \big | \le \mathrm{C} \alpha, \quad \big | \nabla E_\sigma (P) - P \big | \le \mathrm{C} \alpha, \\
& \big | \nabla E_\sigma(P) - \nabla E_\sigma(P') \big | \le \mathrm{C} | P - P' |\quad \text{for all } |P| , |P'| \le p_c, \label{eq:nablaE(P)-nablaE(P')}
\end{align}
where $\mathrm{C}$ is a positive constant independent of $\sigma$.
\item[4)] For all $\sigma \ge 0$, $|P| \le p_c$ and $k \in \mathbb{R}^3$,
\begin{equation}
E_\sigma(P-k) \ge E_\sigma(P) - \frac{1}{3} |k|.
\end{equation}
\end{itemize}
\end{proposition}
We fix $P \in \mathbb{R}^3$ and, to simplify notations, we drop, from now on,  the dependence on $P$ everywhere unless a confusion may arise.
Note that the Hilbert space $\mathcal{F}$ is unitarily equivalent to $\mathcal{F}_\sigma \otimes \mathcal{F}^\sigma$ where $\mathcal{F}^\sigma := \Gamma_s( \mathrm{L}^2 ( \{ (k,\lambda)  , |k| \le \sigma \} ) )$. In this representation we have
\begin{equation}\label{eq:K^sigma=}
H_\sigma = K_\sigma \otimes \mathds{1} + \mathds{1} \otimes \big ( \frac{1}{2} P_f^2 + H_f \big ) - \nabla K_\sigma \otimes P_f,
\end{equation}
where we used (with obvious abuse of notation) that $P_f = P_f \otimes \mathds{1} + \mathds{1} \otimes P_f,\ H_f = H_f \otimes \mathds{1} + \mathds{1} \otimes H_f$ and $A_\sigma = A_\sigma \otimes \mathds{1}$, and where we used the notation
\begin{equation}
\nabla K_\sigma := \nabla H_\sigma |_{\mathcal{F}_\sigma},\quad
\text{with}\quad
\nabla H_\sigma := P - P_f - \alpha^{\frac{1}{2}} A_\sigma.
\end{equation}
In conclusion of this section we mention the decomposition
\begin{equation}
H = H_\sigma + U_\sigma,
\end{equation}
where
\begin{align}
U_\sigma &:= - \alpha^{\frac{1}{2}} \nabla K_\sigma \otimes A^\sigma + \frac{ \alpha^{ \frac{1}{2} } }{2} \mathds{1} \otimes \Big ( A^\sigma \cdot P_f + P_f \cdot A^\sigma \Big ) + \frac{ \alpha }{2} \mathds{1} \otimes ( A^\sigma )^2, \label{eq:Rsigma1}
\end{align}
and
\begin{equation}
A^\sigma := \Phi ( h^\sigma ), \quad h^\sigma( k,\lambda ) := h( k,\lambda ) - h_\sigma( k,\lambda ) = \frac{ \mathds{1}_{ \{ |k| \le \sigma \} }(k) }{ |k|^{\frac{1}{2}} } \varepsilon_\lambda(k).
\end{equation}

\section{Feshbach-Schur operator}\label{section:Feshbach}

In this section we use the ``smooth Feshbach-Schur map'', $F_\chi$, introduced in \cite{BCFS} to map the operators $H - \lambda$ onto more tractable operators. Define
\begin{equation}
\chi^\sigma_f := \chi^\sigma_f(H_f) \equiv \kappa^{\rho\sigma} ( H_f ),\quad \bar \chi^\sigma_f := \sqrt{ \mathds{1} - ( \chi^\sigma_f )^2 },
\end{equation}
with $\kappa^{\rho \sigma}$ as defined in \eqref{eq:kappaLambda}, $\rho$ the same as in \eqref{gapineq}, and
\begin{equation}
\chi := P_\sigma \otimes \chi^\sigma_f,\quad \bar \chi := P_\sigma \otimes \bar \chi^\sigma_f + \bar P_\sigma \otimes \mathds{1},
\end{equation}
where
\begin{equation}
P_\sigma := \mathds{1}_{\{ E_\sigma \}}(K_\sigma)\ \mbox{and}\ \bar P_\sigma:=\mathds{1} - P_\sigma.
\end{equation} Note that $\chi^2 + \bar \chi^2=\mathds{1}$ and $[ \chi , \bar \chi ]=0$.

It is tempting to apply the Feshbach-Schur map $F_\chi$ to $H - \lambda$, the operator $T$ of Appendix \ref{appendix:Feshbach} being chosen as $T=H_\sigma - \lambda$. However this choice is not suitable, since, because of the term $-\nabla K_\sigma \otimes P_f$ in $H_\sigma$ (see Equation \eqref{eq:K^sigma=}), the commutator $[H_\sigma,\chi]$ does not vanish (hence Hypothesis (1) of Appendix \ref{appendix:Feshbach} is not satisfied).

One could apply $F_\chi$ to $H - \lambda$ with $T = H_\sigma + \nabla K_\sigma \otimes P_f - \lambda$. However, as far as the Mourre estimate of Section \ref{section:mourre} is concerned, this choice is not suitable either, since it gives rise to ``perturbation'' terms of order $O(\sigma)$ in $F_\chi( H - \lambda )$, that is the same order as the leading order terms in $F_\chi( H - \lambda )$.

To circumvent this difficulty, we set $T_{\sigma} := H_\sigma + ( \nabla K_\sigma - \nabla E_\sigma) \otimes P_f$, that is
\begin{equation} \label{eq:Hsigma+}
T_\sigma = K_\sigma \otimes \mathds{1} + \mathds{1} \otimes \big ( \frac{1}{2} P_f^2 + H_f \big ) - \nabla E_\sigma \otimes P_f.
\end{equation}
Notice that $[ \chi , T_{\sigma} ] =0$, and that
\begin{equation}\label{eq:Rsigma+}
H = T_\sigma + W_\sigma, \quad \text{where} \quad W_\sigma := U_\sigma - ( \nabla K_\sigma - \nabla E_\sigma ) \otimes P_f.
\end{equation}
Using the Feynman-Hellman formula, we shall see in the following that the term $( \nabla K_\sigma - \nabla E_\sigma ) \otimes P_f$ can indeed be treated as a perturbation, and leads to terms of order $O( \alpha^{1/2} \sigma )$ in $F_\chi( H - \lambda )$; (see Lemmata \ref{lm:nablaK-nablaE}, \ref{lm:||W||} and \ref{lm:[W,iB]}).

On operators of the form $H-\lambda$ we introduce the Feshbach-Schur map (see Appendix \ref{appendix:Feshbach}):
\begin{equation}\label{eq:expr_Feshbach1}
F_\chi( H - \lambda) = T_\sigma - \lambda + \chi W_\sigma \chi - \chi W_\sigma \bar \chi \big [ H_{\bar \chi} - \lambda \big ]^{-1} \bar \chi W_\sigma \chi,
\end{equation}
where (cf. Appendix \ref{appendix:Feshbach})
\begin{align}
& H_{\bar \chi} := T_\sigma + \bar \chi W_{\sigma} \bar \chi. \label{eq:Hbarchi}
\end{align}
This family is well-defined as follows from the fact that the operators $\chi W_\sigma$ and $W_\sigma \chi$ are bounded and from the Proposition \ref{prop:Hchibar} below.
The Feynman-Hellman formula gives $P_\sigma \nabla K_\sigma P_\sigma = \nabla E_\sigma P_\sigma$ and hence $\chi W_\sigma \chi = \chi U_\sigma \chi$. Thus Equations \eqref{eq:Hsigma+} and \eqref{eq:expr_Feshbach1} imply
\begin{align}
F_\chi( H - \lambda ) =& K_\sigma \otimes\mathds{1} + \mathds{1} \otimes \big( \frac{1}{2} P_f^2 + H_f \big ) -  \nabla E_\sigma \otimes P_f  - \lambda \notag \\
& + {\chi U_{\sigma} \chi} - \chi W_\sigma \bar \chi \big [ H_{\bar \chi}-\lambda \big ]^{-1} \bar \chi W_\sigma \chi. \label{eq:expr_Feshbach}
\end{align}
\begin{proposition}\label{prop:Hchibar}
For any $\mathrm{C}_0>0$, there exists $\alpha_0>0$ such that, for all $|P| \le p_c$, $0 \le \alpha \le \alpha_0$ and $0 < \sigma \le \mathrm{C}_0 \alpha^{1/2}$, for all $\lambda \le E_\sigma +  \rho \sigma / 4$, $H_{\bar \chi} - \lambda$ is bounded invertible on $\mathrm{Ran}( \bar \chi )$ and
\begin{align}
& \big \| \bar \chi \big [ H_{\bar \chi} - \lambda \big ]^{-1} \bar \chi \big \| \le \mathrm{C} \sigma^{-1}, \label{eq:Feshbach3} \\
& \big \| \bar \chi \big [ H_{\bar \chi} - \lambda \big ]^{-1} \bar \chi W_\sigma  \chi \big \| \le \mathrm{C}. \label{eq:Feshbach4}
\end{align}
\end{proposition}
\begin{proof} By \eqref{eq:Rsigma+}, the perturbation $W_\sigma$ consists of two terms. As a first step in the proof of Proposition \ref{prop:Hchibar}, we focus on the term $(\nabla K_\sigma - \nabla E_\sigma) \otimes P_f$, which is analyzed in the following lemma.
\begin{lemma}\label{lm:f(Hsigma)}
Let
\begin{equation}\label{eq:Hbarchi1}
H_{\bar \chi}^1 := T_\sigma - \bar \chi ( \nabla K_\sigma - \nabla E_\sigma ) \otimes P_f \bar \chi \,.
\end{equation}
For any $\mathrm{C}_0>0$, there exists $\alpha_0>0$ such that, for all $|P| \le p_c$, $0 \le \alpha \le \alpha_0$ and $0 < \sigma \le \mathrm{C}_0 \alpha^{1/2}$, for all $\lambda \le E_\sigma +  \rho \sigma / 4$, $H_{\bar \chi}^1 - \lambda$ is bounded invertible on $\mathrm{Ran}( \bar \chi )$ and
\begin{align}
& \big \| \bar \chi \big [ H_{\bar \chi}^1 - \lambda \big ]^{-1} \bar \chi \big \| \le \mathrm{C} \sigma^{-1}, \label{eq:Feshbach1} \\
& \big \| \bar \chi \big [ H_{\bar \chi}^1 - \lambda \big ]^{-1} \bar \chi ( \nabla K_\sigma - \nabla E_\sigma ) \otimes P_f  \chi \big \| \le \mathrm{C}. \label{eq:Feshbach2}
\end{align}
\end{lemma}
\begin{proof}
Let $\Phi = \bar \chi \Psi \in D( H_\sigma ) \cap \mathrm{Ran} ( \bar \chi )$, $\| \Phi \|=1$. Let us first prove that
\begin{equation}\label{eq:f(Hsigma)_1}
( \Phi , H_\sigma \Phi ) \ge E_\sigma + \frac{3}{8} \rho \sigma.
\end{equation}
We decompose
\begin{align}\label{eq:Hsigma_decomp}
( \Phi , H_\sigma \Phi ) &= ( \Phi , H_\sigma ( \mathds{1} \otimes \mathds{1}_{H_f \le 3 \rho \sigma /4} ) \Phi ) + ( \Phi , H_\sigma ( \mathds{1} \otimes \mathds{1}_{H_f \ge 3 \rho \sigma /4} ) \Phi ),
\end{align}
and use that $\Phi = \bar \chi \Psi =(\bar P_\sigma \otimes \mathds{1} ) \Psi + ( P_\sigma \otimes \bar{\chi}^\sigma_f ) \Psi$. Using Lemma \ref{lm:appendix1} and the fact that $\mathds{1}_{ H_f \le 3\rho\sigma/4 } \, \bar{\chi}^\sigma_f = 0$, we obtain that
\begin{align}
( \Phi , H_\sigma ( \mathds{1} \otimes \mathds{1}_{H_f \le 3 \rho \sigma /4} ) \Phi ) & \ge (1 - \frac{3}{4} \rho \sigma ) ( \Phi , K_\sigma \otimes \mathds{1} ( \mathds{1} \otimes \mathds{1}_{H_f \le 3 \rho \sigma /4} ) \Phi ) \notag \\
& = ( 1 - \frac{3}{4} \rho \sigma ) ( ( \bar P_\sigma \otimes \mathds{1} ) \Psi , K_\sigma \otimes \mathds{1} ( \bar P_\sigma \otimes \mathds{1}_{H_f \le 3 \rho \sigma /4} ) \Psi ).
\end{align}
Since, by Proposition \ref{prop:Esigma(P)}, $\bar P_\sigma K_\sigma \bar P_\sigma \ge E_\sigma + \rho \sigma$, this implies
\begin{align}
( \Phi , H_\sigma ( \mathds{1} \otimes \mathds{1}_{H_f \le 3 \rho \sigma /4} ) \Phi ) & \ge ( 1 - \frac{3}{4} \rho \sigma ) ( E_\sigma + \rho \sigma ) ( \Phi , ( \mathds{1} \otimes \mathds{1}_{H_f \le 3 \rho \sigma / 4 } ) \Phi ) \notag \\
& \ge ( E_\sigma + \frac{3}{8} \rho \sigma ) ( \Phi , ( \mathds{1} \otimes \mathds{1}_{H_f \le 3 \rho \sigma / 4 } ) \Phi ).
\end{align}
Note that in the last inequality we used that, by Proposition \ref{prop:Esigma(P)}, $E_\sigma \le 1/100$ for $|P| \le 1/40$ and $\alpha$ sufficiently small. The second term on the hand side of \eqref{eq:Hsigma_decomp} is estimated thanks to Lemma \ref{lm:appendix0}, which gives:
\begin{align}
( \Phi , H_\sigma ( \mathds{1} \otimes \mathds{1}_{H_f \ge 3 \rho \sigma /4} ) \Phi ) &\ge E_\sigma + \frac{1}{2} ( \Phi , ( \mathds{1} \otimes H_f ) ( \mathds{1} \otimes \mathds{1}_{H_f \ge 3 \rho \sigma /4} ) \Phi ) \notag \\
&\ge ( E_\sigma + \frac{3}{8} \rho \sigma ) ( \Phi , ( \mathds{1} \otimes \mathds{1}_{H_f \ge 3 \rho \sigma / 4 } ) \Phi ).
\end{align}
Hence \eqref{eq:f(Hsigma)_1} is proven.

From the definition of $H_{\bar \chi}^1$, we infer that
\begin{align}
H_{\bar \chi}^1
=& H_\sigma + (\nabla K_\sigma-\nabla E_\sigma) \otimes P_f
- \bar\chi (\nabla K_\sigma-\nabla E_\sigma) \otimes P_f \bar\chi \notag\\
=& H_\sigma + \big ( P_\sigma \otimes ( \bar\chi^\sigma_f - \mathds{1} ) \big ) \nabla K_\sigma \otimes P_f \big( \bar P_\sigma \otimes \mathds{1} \big ) \notag \\
& + \big( \bar P_\sigma \otimes \mathds{1} \big ) \nabla K_\sigma \otimes P_f \big ( P_\sigma \otimes ( \bar\chi^\sigma_f - \mathds{1} ) \big ) \label{eq:lm_Feshbach2}
\end{align}
where we used that $\bar\chi= P_\sigma \otimes ( \bar\chi^\sigma_f - \mathds{1} )  + \mathds{1}\otimes\mathds{1}$, and
\begin{align}
&	\big(\mathds{1}\otimes\mathds{1}\big)(\nabla K_\sigma-\nabla E_\sigma) \otimes P_f
	\big( P_\sigma \otimes ( \bar\chi^\sigma_f - \mathds{1} ) \big) \notag \\
&	=
	\big( \bar P_\sigma\otimes\mathds{1} \big)
	\nabla K_\sigma\otimes P_f \big( P_\sigma \otimes ( \bar\chi^\sigma_f - \mathds{1} )  \big). \label{eq:FH_formula_2}
\end{align}
Equation \eqref{eq:FH_formula_2} follows from the Feynman-Hellman formula, $P_\sigma \nabla K_\sigma P_\sigma = \nabla E_\sigma P_\sigma$, and orthogonality, $P_\sigma\bar P_\sigma=0$.
By Proposition \ref{prop:Esigma(P)}, for $|P|\le p_c = 1/40$ and $\alpha$ sufficiently small,
\begin{equation}\label{eq:lm_Feshbach3}
\big \| \nabla K_\sigma P_\sigma \big \|^2 \le 2 E_\sigma \le P^2 + \mathrm{C} \alpha
\le \frac{1}{36^2}.
\end{equation}
Thus, combined with
\begin{equation}\label{eq:lm_Feshbach4}
\| P_f ( \bar\chi^\sigma_f - \mathds{1} ) \| \le 2 \| H_f ( \bar\chi^\sigma_f - \mathds{1} ) \| \le 2 \rho \sigma
\end{equation}
and \eqref{eq:f(Hsigma)_1}, Equations \eqref{eq:lm_Feshbach2}--\eqref{eq:lm_Feshbach3} imply that
\begin{equation}
(\Phi , H_{\bar \chi}^1 \Phi) \ge E_\sigma + ( \frac{3}{8} - \frac{1}{9} )  \rho \sigma = E_\sigma + \frac{19}{72}  \rho \sigma,
\end{equation}
provided that $\alpha$ is sufficiently small. This establishes that $H_{\bar \chi}^1 - \lambda$ is bounded invertible for any $\lambda \le E_\sigma +  \rho \sigma / 4$, and leads to \eqref{eq:Feshbach1}. To obtain \eqref{eq:Feshbach2}, it suffices to combine \eqref{eq:Feshbach1} with \eqref{eq:lm_Feshbach3} and the fact that $\|P_f \chi^\sigma_f \| \le \mathrm{C} \sigma$. 
\end{proof}
We return to the proof of Proposition \ref{prop:Hchibar}. Using the operator $H_{\bar\chi}^1$ introduced in the statement of Lemma \ref{lm:f(Hsigma)}, we have that
\begin{equation}
H_{\bar\chi} = H_{\bar\chi}^1 + \bar\chi U_\sigma \bar\chi.
\end{equation}
Consider the Neumann series
\begin{equation}\label{eq:Neumann}
\bar \chi \big [ H_{\bar \chi} - \lambda \big ]^{-1} \bar \chi = \bar \chi \big [ H_{\bar \chi}^1 - \lambda \big ]^{-1} \sum_{n \ge 0} \Big ( - \bar \chi U_\sigma \bar \chi \big [ H_{\bar \chi}^1 - \lambda \big ]^{-1} \Big )^n \bar \chi.
\end{equation}
We claim that
\begin{equation}\label{eq:Neumann_2}
\big \| \big [ H_{\bar \chi}^1 - \lambda \big ]^{ - \frac{1}{2} } \bar \chi U_\sigma \bar \chi \big [ H_{\bar \chi}^1 - \lambda \big ]^{ - \frac{1}{2} } \bar \chi \big \| \le \mathrm{C} \alpha^{\frac{1}{2}}.
\end{equation}
Indeed, inserting the expression \eqref{eq:Rsigma1} of $U_\sigma$ into the left-hand side of \eqref{eq:Neumann_2}, we obtain three terms: The first one is given by
\begin{equation}\label{eq:Neumann_3}
\big \| \alpha^{\frac{1}{2}} \big [ H_{\bar \chi}^1 - \lambda \big ]^{ - \frac{1}{2} } \bar \chi \nabla K_\sigma \otimes A^\sigma \bar \chi \big [ H_{\bar \chi}^1 - \lambda \big ]^{ - \frac{1}{2} } \bar \chi \big \|.
\end{equation}
It follows from Lemmata \ref{lm:standard1}, \ref{lm:appendix0} and \ref{lm:f(Hsigma)} that
\begin{equation}
\big \| ( \mathds{1} \otimes a(h^\sigma) ) \bar \chi \big [ H_{\bar \chi}^1 - \lambda \big ]^{ - \frac{1}{2} } \bar \chi \big \| \le \mathrm{C} \sigma^{\frac{1}{2}}.
\end{equation}
Using in addition that, by Lemma \ref{lm:f(Hsigma)},
\begin{equation}
\big \| ( \nabla K_\sigma \otimes \mathds{1} ) \bar \chi \big [ H_{\bar \chi}^1 - \lambda \big ]^{ - \frac{1}{2} } \bar \chi \big \| \le \mathrm{C} \sigma^{ - \frac{1}{2}},
\end{equation}
we get $\| \eqref{eq:Neumann_3} \| \le \mathrm{C} \alpha^{1/2}$. The second and third terms from $\eqref{eq:Rsigma1}$ are estimated similarly, which leads to \eqref{eq:Neumann_2}.
Together with \eqref{eq:Feshbach1} from Lemma \ref{lm:f(Hsigma)} this implies that, for any $n\in\mathbb{N}$,
\begin{equation}
\Big \| \bar \chi \big [ H_{\bar \chi}^1 - \lambda \big ]^{-1} \Big ( - \bar \chi U_\sigma \bar \chi \big [ H_{\bar \chi}^1 - \lambda \big ]^{-1} \Big )^n \bar \chi \Big \| \le \mathrm{C} \sigma^{-1} ( \mathrm{C}' \alpha^{\frac{1}{2}} )^n.
\end{equation}
Hence, for $\alpha$ sufficiently small, the right-hand-side of \eqref{eq:Neumann} is convergent and \eqref{eq:Feshbach3} holds. Estimate \eqref{eq:Feshbach4} follows similarly.
\end{proof}
%
%

\section{Mourre estimate for the Feshbach-Schur operator}\label{section:mourre}

In this section we shall prove Theorem \ref{thm:LAP} in the case where
$$J = J_\sigma^< := [ E(P) + 11 \rho \sigma / 128 , E(P) + 13 \rho \sigma / 128 ],$$
and $\sigma$ is such that $\sigma \le \mathrm{C}_0 \alpha^{1/2}$. We shall begin with proving a limiting absorption principle for the  Feshbach-Schur operator
\begin{equation}\label{eq:defF(z)}
F(\lambda) := F_\chi ( H - \lambda ) |_{ \mathrm{Ran} ( P_\sigma \otimes \mathds{1} )},
\end{equation}
defined in \eqref{eq:expr_Feshbach1}, Section \ref{section:Feshbach}. Note that the operator $F(\lambda) $ is self-adjoint $\forall \lambda \in J_\sigma^<$. Here the parameter $\lambda$ shall be fixed such that $\lambda \in J_\sigma^<$ and we shall prove a LAP for $F(\lambda)$ on the interval $\Delta_\sigma$ defined in this section by
\begin{equation}\label{eq:Deltasigma_2}
\Delta_\sigma = [ - \rho \sigma / 128 , \rho \sigma / 128 ].
\end{equation}
Then we shall deduce a limiting absorption principle for $H$ near the ground state energy $E$ by applying Theorem \ref{thm:LAPtransfer}.

We begin with showing the Mourre estimate for  $F(\lambda)$, $\lambda \in J_\sigma^<$.

Recall that $\kappa^\sigma$ denotes a function in  $\mathrm{C}_0^\infty( \{ k , |k| \le \sigma \} ; [0,1] )$ such that $\kappa^\sigma =1$ on $\{ k , |k| \le 3 \sigma / 4 \}$. The conjugate operator we shall use in this section is the operator $B^\sigma$, defined by:
\begin{align}\label{eq:Bsigma}
& B^\sigma =  \d \Gamma( b^\sigma ), \quad\text{with}\quad b^\sigma = \kappa^\sigma b \kappa^\sigma.
\end{align}
Clearly, $B^\sigma$ acts on the second component of the tensor product $\mathcal{F}_\sigma \otimes \mathcal{F}^\sigma$. The main theorem of this section is:
\begin{theorem}\label{thm:Mourre}
For any $\mathrm{C}_0>0$, there exists $\alpha_0>0$ such that, for all $|P| \le p_c$, $0 \le \alpha \le \alpha_0$, $0 < \sigma \le \mathrm{C}_0 \alpha^{1/2}$, and $\lambda \in J_\sigma^<$,
\begin{equation}
\mathds{1}_{\Delta_\sigma}(F(\lambda)) [ F(\lambda) , \i B^\sigma ] \mathds{1}_{\Delta_\sigma}(F(\lambda)) \ge \frac{\rho \sigma}{128} \mathds{1}_{\Delta_\sigma}(F(\lambda)).
\end{equation}
\end{theorem}
Before proceeding to the proof of this theorem we draw the desired conclusions from it.
\begin{proposition}\label{lm:LAP_F}
For any $\mathrm{C}_0>0$, there exists $\alpha_0>0$ such that, for any $|P| \le p_c$, $0 \le \alpha \le \alpha_0$, $0<\sigma\le \mathrm{C}_0 \alpha^{1/2}$, $1/2 < s \le 1$, and $\lambda \in J_\sigma^<$,
\begin{equation}\label{eq:LAP_F_1}
\sup_{z \in (\Delta_\sigma)_{\pm}} \big \| \langle B^\sigma \rangle^{- s} \big [ F(\lambda) - z \big ]^{-1} \langle B^\sigma \rangle^{-s} \big \| < \infty.
\end{equation}
Here $(\Delta_\sigma)_{\pm} = \{ z \in \mathbb{C} , \mathrm{Re} z \in [ - \rho \sigma / 128 , \rho \sigma / 128 ]  , 0 < \pm \mathrm{Im} z \le 1 \}$. Moreover, the map
\begin{equation}\label{eq:LAP_F_2}
J_\sigma^< \times \Delta_\sigma \ni ( \lambda \times \mu ) \mapsto \langle B^\sigma \rangle^{- s} \big [ F(\lambda) - \mu \pm \i 0 \big ]^{-1} \langle B^\sigma \rangle^{-s} \in B( \mathcal{H} )
\end{equation}
is uniformly H\"older continuous in $( \lambda , \mu )$ of order $s-1/2$.
\end{proposition}
\begin{proof}
It follows from Equations \eqref{eq:Hsigma+} and \eqref{eq:expr_Feshbach1} that
\begin{align}
F(\lambda) =& \mathds{1} \otimes \big ( \frac{1}{2} P_f^2 + H_f \big ) - \nabla E_\sigma \otimes P_f + E_\sigma- \lambda, \notag \\
& + \chi W_\sigma \chi - \chi W_\sigma \bar \chi \big [ H_{\bar \chi} - \lambda \big ]^{-1} \bar \chi W_\sigma \chi.
\end{align}
By standard Mourre theory (see for instance \cite{ABG}) and in view of Theorem \ref{thm:Mourre}, 
the limiting absorption principle \eqref{eq:LAP_F_1} and the H{\"o}lder continuity in $\mu$ follow from the fact that $F(\lambda) \in \mathrm{C}^2( B^\sigma )$. Since $\chi W_\sigma$ and $W_\sigma \chi$ are bounded operators, it follows that $D( F(\lambda) ) = D( \mathds{1} \otimes ( \frac{1}{2} P_f^2 + H_f ) )$, and, using the method of \cite[Proposition 9]{FGS1}, one verifies that
\begin{equation}
e^{ \i t B^\sigma } D( \mathds{1} \otimes ( \frac{1}{2} P_f^2 + H_f ) ) \subset D( \mathds{1} \otimes ( \frac{1}{2} P_f^2 + H_f ) ),
\end{equation}
for all $t \in \mathbb{R}$. Hence it suffices to show that $[ F(\lambda) , \i B^\sigma ]$ and $[ [ F(\lambda) , \i B^\sigma ] , \i B^\sigma ]$ are bounded with respect to $\mathds{1} \otimes ( \frac{1}{2} P_f^2 + H_f )$, which follows easily from the expressions of the commutators; (see, in particular, the proofs of Lemmata \ref{lm:[F0,iB]} and \ref{lm:[W,iB]}).

Now, for $\lambda, \lambda' \in J_\sigma^<$, we have
\begin{equation}
F(\lambda) - F(\lambda') = ( \lambda' - \lambda ) \big ( P_\sigma \otimes \mathds{1} + \chi W_\sigma \bar \chi \big [ H_{\bar \chi} - \lambda \big ]^{-1} \big [ H_{\bar \chi} - \lambda' \big ]^{-1} \bar \chi W_\sigma \chi \big ).
\end{equation}
Equation \eqref{eq:Feshbach4} in the statement of Proposition \ref{prop:Hchibar} implies that
\begin{equation}
\big \| \chi W_\sigma \bar \chi \big [ H_{\bar \chi} - \lambda \big ]^{-1} \big [ H_{\bar \chi} - \lambda' \big ]^{-1} \bar \chi W_\sigma \chi \big \| \le \mathrm{C},
\end{equation}
where $\mathrm{C}$ is independent of $\lambda$ and $\lambda'$. Thus, the H{\"o}lder continuity in $(\lambda,\mu)$ stated in \eqref{eq:LAP_F_2} follows again by standard arguments of Mourre theory (see \cite{PSS,AHS,HS}).
\end{proof}
This proposition and Theorem \ref{thm:LAPtransfer} imply the following
\begin{corollary}\label{cor:LAP_K}
For any $\mathrm{C}_0>0$, there exists $\alpha_0>0$ such that, for any $|P| \le p_c$, $0 \le \alpha \le \alpha_0$, $0<\sigma\le \mathrm{C}_0 \alpha^{1/2}$ and $1/2 < s \le 1$,
\begin{equation}
\sup_{z \in (J_\sigma^<)_{\pm}} \big \| \langle B^\sigma \rangle^{- s} \big [ H - z \big ]^{-1} \langle B^\sigma \rangle^{-s} \big \| < \infty,
\end{equation}
where $(J_\sigma^<)_{\pm} = \{ z \in \mathbb{C} , \mathrm{Re}z \in [ E + 11 \rho \sigma / 128 , E + 13 \rho \sigma / 128 ] , 0 < \pm \mathrm{Im}z \le 1 \}$. Moreover, the map
\begin{equation}
[ E + \frac{ 11 \rho \sigma }{128} , E + \frac{ 13 \rho \sigma }{ 128 } ]  \ni \lambda \mapsto \langle B^\sigma \rangle^{- s} \big [ H - \lambda \pm \i 0 \big ]^{-1} \langle B^\sigma \rangle^{-s} \in B( \mathcal{H} )
\end{equation}
is uniformly H\"older continuous in $\lambda$ of order $s-1/2$.
\end{corollary}
By arguments similar to ones used in the proof of Corollary \ref{cor:LAPaway_2}, Corollary \ref{cor:LAP_K} implies the following result.
\begin{corollary}\label{cor:LAP_K_2}
Under the conditions of Corollary \ref{cor:LAP_K}, for any compact interval $J \subset (E , \mathrm{C}_0 \alpha^{\frac{1}{2}}]$,
\begin{equation}
\sup_{z \in J_{\pm}} \big \| ( \d \Gamma (\langle y \rangle) + 1)^{- s} \big [ H(P) - z \big ]^{-1} ( \d \Gamma (\langle y \rangle) + 1)^{-s} \big \| < \infty,
\end{equation}
and the map
\begin{equation}
J \ni \lambda \mapsto ( \d \Gamma (\langle y \rangle) + 1)^{- s} \big [ H(P) - \lambda \pm \i 0 \big ]^{-1} ( \d \Gamma (\langle y \rangle) + 1)^{-s} \in B( \mathcal{H} )
\end{equation}
is uniformly H\"older continuous in $\lambda$ of order $s-1/2$.
\end{corollary}

Now we proceed to the proof of Theorem \ref{thm:Mourre}. It will be divided into a sequence of Lemmata. In what follows we often do not display the argument $\lambda$. First, let us write

\begin{equation}
F = F_0 + W_1 + W_2,
\end{equation}
where
\begin{align}
& F_0 := \mathds{1} \otimes \big ( \frac{1}{2} P_f^2 + H_f \big ) - \nabla E_\sigma \otimes P_f + E_\sigma- \lambda, \\
& W_1 := \chi U_\sigma \chi, \  ( = \chi W_\sigma \chi \mbox{ by Feynman-Hellman; see above}) \label{eq:def_W1} \\
& W_2 := - \chi W_\sigma \bar \chi \big [ H_{\bar \chi} - \lambda \big ]^{-1} \bar \chi W_\sigma \chi. \label{eq:def_W2}
\end{align}
Let us begin by estimating $[F_0,\i B^\sigma]$ from below on the range of $\mathds{1} \otimes \mathds{1}_{ H_f \le \delta \rho \sigma }$, for some suitably chosen $\delta > 0$.
\begin{lemma}\label{lm:[F0,iB]}
Let $|P| \le p_c$ and $\delta > 0$ be such that $\delta \rho \sigma < 3 \sigma / 4$. Then on $\mathrm{Ran} ( \mathds{1} \otimes \mathds{1}_{ H_f \le \delta \rho \sigma } )$,
\begin{align}
& \big [ F_0, \i B^\sigma \big ]  \ge \frac{1}{2} ( \mathds{1} \otimes H_f )  - \mathrm{C} \sigma^2 ,
\end{align}
where $\mathrm{C}$ is a positive constant.
\end{lemma}
\begin{proof}
We have that
\begin{align}\label{eq:[Hf,iB]}
[ H_f , \i B^\sigma ] = \d \Gamma ( \kappa^\sigma (k)^2 |k| ), \quad [ P_f , \i B^\sigma ] = \d \Gamma ( \kappa^\sigma(k)^2 k ).
\end{align}
Therefore,
\begin{align}\label{eq:[F0,iB]}
\big [ F_0, \i B^\sigma \big ] = \mathds{1} \otimes (P_f \cdot \d \Gamma ( \kappa^\sigma(k)^2 k )+  \d \Gamma ( \kappa^\sigma (k)^2 |k| )) - \nabla E_\sigma \otimes \d \Gamma ( \kappa^\sigma(k)^2 k ).
\end{align}
For $j=1,2,3$, we have
\begin{equation}\label{eq:[Hf,iB]_2}
\pm \d \Gamma ( \kappa^\sigma(k)^2 k_j ) \le \d \Gamma ( \kappa^\sigma(k)^2 | k | ) \le \mathds{1} \otimes H_f,
\end{equation}
so that
\begin{align}
\nabla E_\sigma \otimes \d \Gamma ( \kappa^\sigma(k)^2 k ) & \ge - ( \sum_j | (\nabla E_\sigma)_j | ) \d \Gamma ( \kappa^\sigma(k)^2 | k | ) \notag \\
& \ge - 2 | \nabla E_\sigma | \d \Gamma ( \kappa^\sigma(k)^2 | k | ). \label{eq:[nablaE_Pf,iB]}
\end{align}
Moreover, using again \eqref{eq:[Hf,iB]_2}, it can easily be checked that
\begin{equation}\label{eq:[Pf2,iB]}
\mathds{1} \otimes \big ( P_f \cdot \d \Gamma ( \kappa^\sigma(k)^2 k ) \mathds{1}_{ H_f \le \delta \rho \sigma } \big ) \ge - \mathrm{C} \sigma^2.
\end{equation}
Hence Equations \eqref{eq:[F0,iB]}, \eqref{eq:[nablaE_Pf,iB]} and \eqref{eq:[Pf2,iB]} yield
\begin{align}
& \big [ F_0, \i B^\sigma \big ] ( \mathds{1} \otimes \mathds{1}_{ H_f \le \delta \rho \sigma } ) \notag \\
& \ge ( 1 - 2 | \nabla E_\sigma | ) ( \mathds{1} \otimes \d \Gamma ( \kappa^\sigma (k)^2 |k| ) ) ( \mathds{1} \otimes \mathds{1}_{ H_f \le \delta \rho \sigma } ) - \mathrm{C} \sigma^2 ( \mathds{1} \otimes \mathds{1}_{ H_f \le \delta \rho \sigma } ) \notag \\
& \ge \frac{1}{2} ( \mathds{1} \otimes \d \Gamma ( \kappa^\sigma (k)^2 |k| ) ) ( \mathds{1} \otimes \mathds{1}_{ H_f \le \delta \rho \sigma } ) - \mathrm{C} \sigma^2 ( \mathds{1} \otimes \mathds{1}_{ H_f \le \delta \rho \sigma } ). \label{eq:Hfsigma/2_1}
\end{align}
In the second inequality we used that, by Proposition \ref{prop:Esigma(P)}, $| \nabla E_\sigma | \le |P| + \mathrm{C} \alpha^{\frac{1}{2}} \le 1/4$ for $|P| \le 1/40$ and $\alpha$ sufficiently small. To conclude the proof of the lemma, it remains to justify that the operator $\d \Gamma ( \kappa^\sigma (k)^2 |k| )$ in \eqref{eq:Hfsigma/2_1} can be replaced by $H_f$. To this end, we define
\begin{align}
& H_{f,3\sigma/4}^\sigma = \sum_{\lambda=1,2} \int_{ 3 \sigma / 4 \le |k| \le \sigma } |k| a^*_\lambda( k) a_\lambda(k) \d k, \notag \\
& N_{3\sigma/4}^\sigma = \sum_{\lambda=1,2} \int_{ 3 \sigma / 4 \le |k| \le \sigma } a^*_\lambda(k) a_\lambda(k) \d k,
\end{align}
and $P_{3\sigma/4}^\sigma = \mathds{1}_{ \{ 0 \} } (H_{f,3\sigma/4}^\sigma)$, $\bar P_{3\sigma/4}^\sigma = \mathds{1} - P_{3\sigma/4}^\sigma$. Then we have that
\begin{equation}
( \mathds{1} \otimes H_f ) \bar P_{3\sigma/4}^\sigma \ge H_{f,3\sigma/4}^\sigma \bar P_{3\sigma/4}^\sigma \ge \frac{3\sigma}{4} N_{3\sigma/4}^\sigma \bar P_{3\sigma/4}^\sigma \ge \frac{3\sigma}{4} \bar P_{3\sigma/4}^\sigma.
\end{equation}
Therefore, since $ \mathds{1} \otimes H_f $ commutes with $P_{3\sigma/4}^\sigma$, we get
\begin{align}
\delta \rho \sigma \bar P_{3\sigma/4}^\sigma ( \mathds{1} \otimes \mathds{1}_{H_f \le \delta \rho \sigma} ) & \ge ( \mathds{1} \otimes H_f ) \bar P_{3\sigma/4}^\sigma ( \mathds{1} \otimes \mathds{1}_{H_f \le \delta \rho \sigma} ) \notag \\
& \ge \frac{ 3 \sigma }{4} \bar P_{3\sigma/4}^\sigma ( \mathds{1} \otimes \mathds{1}_{H_f \le \delta \rho \sigma} )
\end{align}
and since $\delta \rho \sigma < 3\sigma / 4$ by assumption, this implies
\begin{equation}\label{eq:Hfsigma/2_2}
( \mathds{1} \otimes \mathds{1}_{H_f \le \delta \rho \sigma} ) = P_{3\sigma/4}^\sigma ( \mathds{1} \otimes \mathds{1}_{H_f \le \delta \rho \sigma} ).
\end{equation}
Since $\kappa^\sigma(k)=1$ for any $|k| \le 3\sigma/4$, we obtain that
\begin{equation}\label{eq:Hfsigma/2_3}
\big ( \mathds{1} \otimes \d \Gamma ( \kappa^\sigma (k)^2 |k| ) \big ) P_{3 \sigma / 4 }^\sigma  = ( \mathds{1} \otimes H_f ) P_{3\sigma/4}^\sigma.
\end{equation}
We conclude the proof using \eqref{eq:Hfsigma/2_1}, \eqref{eq:Hfsigma/2_2}, \eqref{eq:Hfsigma/2_3}, and the fact that
\begin{equation}
\mathds{1} \otimes \d \Gamma ( \kappa^\sigma (k)^2 |k| ) \ge \big ( \mathds{1} \otimes \d \Gamma ( \kappa^\sigma (k)^2 |k| ) \big ) P_{3 \sigma / 4 }^\sigma.
\end{equation}
\end{proof}
The following lemma is an important ingredient to show Theorem \ref{thm:Mourre}. It justifies the fact that one can consider the term $(\nabla K_\sigma - \nabla E_\sigma ) \otimes P_f$ in $W_\sigma$ as a small perturbation. The idea of its proof is due to \cite{AFGG}, and is based on the $\mathrm{C}^2$-regularity of the map $P \mapsto E_\sigma(P)$ uniformly in $\sigma$, established in \cite{Chen} and \cite{FP} (see more precisely inequality \eqref{eq:nablaE(P)-nablaE(P')} in Proposition \ref{prop:Esigma(P)}).

Let $(e_j$, $j=1,2,3$) be the canonical orthonormal basis of $\mathbb{R}^3$. For any $y \in \mathbb{R}^3$, we set $y_j = y \cdot e_j$.
\begin{lemma}\label{lm:nablaK-nablaE}
For any $\mathrm{C}_0>0$, there exists $\alpha_0>0$ such that, for all $|P| \le p_c$, $0 \le \alpha \le \alpha_0$, $0 < \sigma \le \mathrm{C}_0 \alpha^{1/2}$, $\lambda \in J_\sigma^<$, $j \in \{ 1,2,3 \}$, and $0 < \delta \ll 1$, 
\begin{equation}\label{eq:nablaK-,ablaE}
\Big \| \big [ H_{\bar \chi} - \lambda \big ]^{-\frac{1}{2}} \bar \chi \big ( ( \nabla K_\sigma - \nabla E_\sigma )_j P_\sigma \big ) \otimes \mathds{1}_{H_f \le \delta} \Big \| \le \mathrm{C} \big ( 1 + \delta^{\frac{1}{2}} \sigma^{-\frac{1}{2}} \big ).
\end{equation}
\end{lemma}
\begin{proof}
For any $u>0$, we can write
\begin{equation}\label{eq:nablaK-nablaE_j_0}
( \nabla K_\sigma )_j = \frac{1}{u} \big ( K_\sigma( P + u e_j ) - K_\sigma ( P ) \big ) - \frac{u}{2}.
\end{equation}
Using that $K_\sigma(P) P_\sigma = E_\sigma(P) P_\sigma$, this implies
\begin{align}
( \nabla K_\sigma - \nabla E_\sigma )_j P_\sigma =& \frac{1}{u} ( K_\sigma( P + ue_j ) - E_\sigma(P+ue_j) ) P_\sigma \notag \\
&+ \Big ( \frac{1}{u} ( E_\sigma(P+ue_j) - E_\sigma(P) ) - ( \nabla E_\sigma )_j - \frac{u}{2} \Big ) P_\sigma. \label{eq:nablaK-nablaE_j}
\end{align}
By Proposition \ref{prop:Esigma(P)},
\begin{equation}\label{eq:nablaE_j}
\big | \frac{1}{u} ( E_\sigma(P+ue_j) - E_\sigma(P) ) - ( \nabla E_\sigma )_j \big | \le \mathrm{C} u,
\end{equation}
where $\mathrm{C}$ is independent of $\sigma$. Consequently, it follows from the Feynman-Hellman formula, $P_\sigma (\nabla K_\sigma )_j P_\sigma = (\nabla E_\sigma)_j P_\sigma$, together with Equation \eqref{eq:nablaK-nablaE_j_0} that, for any $\Phi \in \mathrm{Ran}( P_\sigma )$, $\| \Phi \| = 1$,
\begin{align}
& \big \| ( K_\sigma( P + ue_j ) - E_\sigma(P+ue_j) )^{\frac{1}{2}} \Phi \big \|^2 \phantom{ \frac{u^2}{2} } \notag \\
& = \big ( \Phi , ( K_\sigma( P + ue_j ) - E_\sigma(P+ue_j) ) \Phi \big ) \phantom{ \frac{u^2}{2} } \notag \\
& = \big ( \Phi , ( K_\sigma(P) + u ( \nabla K_\sigma )_j + \frac{u^2}{2} - E_\sigma(P+ue_j) ) \Phi \big ) \notag \\
& = E_\sigma(P) - E_\sigma(P+ue_j) + u ( \nabla E_\sigma )_j + \frac{u^2}{2} \le \mathrm{C} u^2. \label{eq:nablaK-nablaE_j_2}
\end{align}
From \eqref{eq:nablaK-nablaE_j}, we obtain that
\begin{align}
( \nabla K_\sigma - \nabla E_\sigma )_j P_\sigma = ( K_\sigma( P + ue_j ) - E_\sigma(P+ue_j) )^{\frac{1}{2}} B_1 + B_2,
\end{align}
where
\begin{align}
& B_1 := \frac{1}{u} ( K_\sigma( P + ue_j ) - E_\sigma(P+ue_j) )^{\frac{1}{2}} P_\sigma, \\
& B_2 := \Big ( \frac{1}{u} ( E_\sigma(P+ue_j) - E_\sigma(P) ) - ( \nabla E_\sigma )_j - \frac{u}{2} \Big ) P_\sigma.
\end{align}
By \eqref{eq:nablaK-nablaE_j_2} and \eqref{eq:nablaE_j}, the operators $B_1,B_2$ are bounded and satisfy
\begin{equation}
\| B_1 \| \le \mathrm{C}, \quad \| B_2 \| \le \mathrm{C} u.
\end{equation}
Thus, choosing $u \le \sigma$, the lemma will follow if we show that
\begin{equation}\label{eq:nablaK-nablaE2}
\Big \| \bar \chi \big [ H_{\bar \chi} - \lambda \big ]^{-\frac{1}{2}} \bar \chi ( K_\sigma( P + ue_j ) - E_\sigma(P+ue_j) )^{\frac{1}{2}} \otimes \mathds{1}_{H_f \le \delta} \Big \|^2 \le \mathrm{C} \delta \sigma^{-1}.
\end{equation}
Let us prove \eqref{eq:nablaK-nablaE2}. To simplify notations, we set
\begin{equation}
\bar \chi_{ \le \delta } := ( \mathds{1} \otimes \mathds{1}_{H_f \le \delta} ) \bar \chi
\end{equation}
Let $\Phi \in \mathrm{Ran}( \bar \chi )$, $\| \Phi \|=1$. Since
\begin{equation}
\big \| \big ( H_{\bar \chi}^1 - \lambda \big ) \big [ H_{\bar \chi} - \lambda \big ]^{-1} \bar \chi \big \| \le \mathrm{C},
\end{equation}
(see the proof of Proposition \ref{prop:Hchibar}), it suffices to estimate
\begin{equation}\label{eq:nablaK-nablaE3}
\Big ( \Phi , \bar \chi \big [ H_{\bar \chi}^1 - \lambda \big ]^{-\frac{1}{2}} \bar \chi_{\le \delta} \big ( ( K_\sigma( P + ue_j ) - E_\sigma(P+ue_j) ) \otimes \mathds{1} \big ) \bar \chi_{\le \delta} \big [ H_{\bar \chi}^1 - \lambda \big ]^{-\frac{1}{2}} \bar \chi \Phi \Big ).
\end{equation}
Using that
\begin{equation}
\Big \| \bar \chi \big [ H_{\bar \chi}^1 - \lambda \big ]^{-\frac{1}{2}} \bar \chi \big ( ( \nabla K_\sigma - \nabla E_\sigma ) \otimes \mathds{1} \big ) \bar \chi \big [ H_{\bar \chi}^1 - \lambda \big ]^{-\frac{1}{2}} \bar \chi \Big \| \le \mathrm{C} \sigma^{-1},
\end{equation}
and since $0 < u \le \sigma$, we get
\begin{equation}\label{eq:nablaK-nablaE4}
\eqref{eq:nablaK-nablaE3} \le \Big ( \Phi , \bar \chi \big [ H_{\bar \chi}^1 - \lambda \big ]^{-\frac{1}{2}} \bar \chi_{\le \delta} \big ( ( K_\sigma( P) - E_\sigma(P) ) \otimes \mathds{1} \big ) \bar \chi_{\le \delta} \big [ H_{\bar \chi}^1 - \lambda \big ]^{-\frac{1}{2}} \bar \chi \Phi \Big ) + \mathrm{C}.
\end{equation}
Next, by Lemma \ref{lm:appendix1},
\begin{align}
& \bar \chi_{\le \delta } \big ( ( K_\sigma( P) - E_\sigma(P) ) \otimes \mathds{1} \big ) \bar \chi_{\le \delta } \notag \\
& \le \frac{1}{1-\delta} \bar \chi_{\le \delta} \big ( \big ( H_\sigma( P) - E_\sigma(P) \big ) + 4 \delta E_\sigma \big ) \bar \chi_{\le \delta}. \label{eq:nablaK-nablaE5}
\end{align}
Using the expression \eqref{eq:lm_Feshbach2} of $H_{\bar\chi}^1$, we conclude from \eqref{eq:nablaK-nablaE5} that
\begin{align}
& \bar \chi_{\le \delta } \big ( ( K_\sigma( P) - E_\sigma(P) ) \otimes \mathds{1} \big ) \bar \chi_{\le \delta } \notag \\
& \le \bar \chi_{\le \delta} \big ( \big ( H_{\bar \chi}^1(P) - E_\sigma(P) \big ) + \mathrm{C} (\sigma + \delta ) \big ) \bar \chi_{\le \delta}. \label{eq:nablaK-nablaE6}
\end{align}
The statement of the lemma follows from \eqref{eq:nablaK-nablaE4}, \eqref{eq:nablaK-nablaE6} and Lemma \ref{lm:f(Hsigma)}.
\end{proof}
In the following lemma, we prove that the ``perturbation'' operators $W_1$, $W_2$ in \eqref{eq:def_W1}--\eqref{eq:def_W2} are of order $O( \alpha^{1/2} \sigma )$.
\begin{lemma}\label{lm:||W||}
For any $\mathrm{C}_0>0$, there exists $\alpha_0>0$ such that, for all $|P| \le p_c$, $0 \le \alpha \le \alpha_0$, $0 < \sigma \le \mathrm{C}_0 \alpha^{1/2}$, and $\lambda \in J_\sigma^<$,
\begin{equation}\label{eq:||W_i||}
\big \| W_i \big \| \le \mathrm{C} \alpha^{\frac{1}{2}} \sigma,\ i=1,2,
\end{equation}
where $W_1$ and $W_2$ are as in \eqref{eq:def_W1}, \eqref{eq:def_W2}.
\end{lemma}
\begin{proof}
Let us first prove \eqref{eq:||W_i||} for $i=1$. Equation \eqref{eq:Rsigma1} combined with the Feynman-Hellman formula gives
\begin{align}
\chi U_\sigma \chi = & - \alpha^{\frac{1}{2}} \big ( \nabla E_\sigma P_\sigma \big ) \otimes \big ( \chi^\sigma_f A^\sigma \chi^\sigma_f \big ) + \frac{ \alpha^{\frac{1}{2}} }{2} P_\sigma \otimes \Big ( \chi^\sigma_f \big ( P_f \cdot A^\sigma + A^\sigma \cdot P_f \big ) \chi^\sigma_f \Big ) \notag \\
& + \frac{ \alpha }{2} P_\sigma \otimes \Big ( \chi^\sigma_f ( A^\sigma )^2 \chi^\sigma_f \Big ). \label{eq:||W_1||_1}
\end{align}
It follows from Lemma \ref{lm:standard1} that
\begin{align}
& \big \| A^\sigma \chi^\sigma_f \big \| \le \mathrm{C} \sigma^{\frac{1}{2}} \big \| [ H_f + \sigma ]^{\frac{1}{2}} \chi^\sigma_f \big \| \le \mathrm{C}' \sigma, \label{eq:||W_1||_2} \\
& \big \| ( A^\sigma \cdot P_f ) \chi^\sigma_f \big \| \le \mathrm{C} \sigma^{\frac{1}{2}} \big \| [ H_f + \sigma ]^{\frac{1}{2}} | P_f | \chi^\sigma_f \big \| \le \mathrm{C}' \sigma^2. \label{eq:||W_1||_2'}
\end{align}
Therefore \eqref{eq:||W_i||} for $i=1$ follows.

To prove \eqref{eq:||W_i||} for $i=2$ it suffices to show that for $\lambda \in J_\sigma^<$,
\begin{equation}\label{eq:||W_2||_1}
\Big \| \big [ H_{\bar \chi} - \lambda \big ]^{-\frac{1}{2}} \bar \chi W_\sigma \chi \Big \| \le \mathrm{C} \alpha^{\frac{1}{4}} \sigma^{\frac{1}{2}}.
\end{equation}
By Equations \eqref{eq:Rsigma1} and \eqref{eq:Rsigma+},
\begin{align}
W_\sigma \chi = & - \alpha^{ \frac{1}{2} } \big ( \nabla K_\sigma P_\sigma \big ) \otimes \big ( A^\sigma \chi^\sigma_f \big ) \label{eq:||W_2||_2}  \\
& + \frac{ \alpha^{\frac{1}{2}} }{2} P_\sigma \otimes \Big ( \big ( P_f \cdot A^\sigma + A^\sigma \cdot P_f \big ) \chi^\sigma_f \Big ) \label{eq:||W_2||_3} \\
& + \frac{ \alpha }{2} P_\sigma \otimes \big ( ( A^\sigma )^2 \chi^\sigma_f \big ) \label{eq:||W_2||_6} \\
& - \Big ( \big ( \nabla K_\sigma - \nabla E_\sigma \big ) P_\sigma \Big ) \otimes \big ( P_f \chi^\sigma_f \big ). \label{eq:||W_2||_7}
\end{align}
We insert this expression into \eqref{eq:||W_2||_1} and estimate each term separately. First, it follows from Proposition \ref{prop:Hchibar} and Estimate \eqref{eq:||W_1||_2} that
\begin{equation}\label{eq:||W_2||_8}
\Big \| \big [ H_{\bar \chi} - \lambda \big ]^{-\frac{1}{2}} \bar \chi \eqref{eq:||W_2||_2} \Big \| \le \mathrm{C} \alpha^{\frac{1}{2}} \sigma^{\frac{1}{2}}.
\end{equation}
Similarly, Lemma \ref{lm:standard2} combined with Proposition \ref{prop:Hchibar} and \eqref{eq:||W_1||_2}--\eqref{eq:||W_1||_2'} implies
\begin{equation}\label{eq:||W_2||_9}
\Big \| \big [ H_{\bar \chi} - \lambda \big ]^{-\frac{1}{2}} \bar \chi \big ( \eqref{eq:||W_2||_3} + \eqref{eq:||W_2||_6} \big ) \Big \| \le \mathrm{C} \alpha^{\frac{1}{2}} \sigma^{\frac{3}{2}}.
\end{equation}
Finally the contribution from \eqref{eq:||W_2||_7} is estimated thanks to Lemma \ref{lm:nablaK-nablaE}: Using \eqref{eq:nablaK-,ablaE} with $\delta = \rho \sigma$, we get, for $j \in \{ 1,2,3 \}$, 
\begin{equation}
\Big \| \big [ H_{\bar \chi} - \lambda \big ]^{-\frac{1}{2}} \bar \chi \Big ( \big ( \nabla K_\sigma - \nabla E_\sigma \big )_j P_\sigma \Big ) \otimes \mathds{1}_{H_f \le \rho \sigma} \Big \| \le \mathrm{C}.
\end{equation}
Together with $\| (P_f)_j \chi^\sigma_f \| \le \mathrm{C} \sigma$, this yields
\begin{equation}\label{eq:||W_2||_13}
\Big \| \big [ H_{\bar \chi} - \lambda \big ]^{-\frac{1}{2}} \bar \chi \eqref{eq:||W_2||_7}  \Big \| \le \mathrm{C} \sigma \le \mathrm{C}' \alpha^{\frac{1}{4}} \sigma^{\frac{1}{2}}.
\end{equation}
Estimates \eqref{eq:||W_2||_8}, \eqref{eq:||W_2||_9} and \eqref{eq:||W_2||_13} imply \eqref{eq:||W_2||_1}, so \eqref{eq:||W_i||}, $i=2$, follows.
\end{proof}
In the next lemma, we estimate the commutators $[W_i, \i B^\sigma ]$, $i=1,2$.
\begin{lemma}\label{lm:[W,iB]}
For any $\mathrm{C}_0>0$, there exists $\alpha_0>0$ such that, for all $|P| \le p_c$, $0 \le \alpha \le \alpha_0$, $0 < \sigma \le \mathrm{C}_0 \alpha^{1/2}$, and $\lambda \in J_\sigma^<$
\begin{align}
& \| [ W_i , \i B^\sigma ] \| \le \mathrm{C} \alpha^{\frac{1}{2}} \sigma,\ i=1, 2, \label{eq:[W_i,iB]}
\end{align}
where $W_1$ and $W_2$ are as in \eqref{eq:def_W1}, \eqref{eq:def_W2}.
\end{lemma}
\begin{proof}
Using for instance the Helffer-Sj{\"o}strand functional calculus, the following identities follow straightforwardly from \eqref{eq:[Hf,iB]}:
\begin{align}
& [ \chi , \i B^\sigma ] = P_\sigma \otimes \big ( \d \Gamma ( \kappa^\sigma(k)^2 |k| ) ( \chi^\sigma_f )' (H_f) \big ), \label{eq:[chi,iB]} \\
& [ \bar \chi , \i B^\sigma ] = P_\sigma \otimes \big ( \d \Gamma ( \kappa^\sigma(k)^2 |k| ) ( \bar \chi^\sigma_f )' (H_f) \big ). \label{eq:[barchi,iB]}
\end{align}
Furthermore,
\begin{equation}\label{eq:[Phi,iB]}
[ A^\sigma , \i B^\sigma ] = - \Phi ( \i b^\sigma h^\sigma ).
\end{equation}
We first prove \eqref{eq:[W_i,iB]} for $i=1$. We have that
\begin{align}
[ W_1 , \i B^\sigma ] = [ \chi , \i B^\sigma ] U_\sigma \chi + \chi [ U_\sigma , \i B^\sigma ] \chi + \chi U_\sigma [ \chi , \i B^\sigma ].
\end{align}
As in the proof of \eqref{eq:||W_i||}, $i=1$, in Lemma \ref{lm:||W||}, we obtain, using \eqref{eq:[chi,iB]}, that
\begin{equation}
\big \| [ \chi , \i B^\sigma ] U_\sigma \chi \big \| = \big \| \chi U_\sigma [ \chi , \i B^\sigma ] \big \| \le \mathrm{C} \alpha^{\frac{1}{2}} \sigma.
\end{equation}
It follows from \eqref{eq:[Hf,iB]} and \eqref{eq:[Phi,iB]} that
\begin{align}
[U_\sigma , \i B^\sigma ] =& \alpha^{ \frac{1}{2} } \nabla K_\sigma \otimes \Phi ( \i b^\sigma h^\sigma ) - \frac{ \alpha^{ \frac{1}{2} } }{2} \mathds{1} \otimes \big ( \Phi( \i b^\sigma h^\sigma ) \cdot P_f + P_f \cdot \Phi ( \i b^\sigma h^\sigma ) \big ) \notag \\
& + \frac{ \alpha^{ \frac{1}{2} } }{2} \mathds{1} \otimes \big ( \Phi(h^\sigma) \cdot \d \Gamma( \kappa^\sigma(k)^2 k ) + \d \Gamma( \kappa^\sigma(k)^2 k ) \cdot \Phi( h^\sigma ) \big ) \notag \\
& - \frac{ \alpha }{2} \mathds{1} \otimes \big ( \Phi ( h^\sigma ) \cdot \Phi( \i b^\sigma h^\sigma ) + \Phi ( \i b^\sigma h^\sigma ) \cdot \Phi ( h^\sigma ) \big ). \phantom{ \frac{ \alpha^{ \frac{1}{2} } }{2} }
\end{align}
Arguing as in the proof of \eqref{eq:||W_i||}, $i=1$, in Lemma \ref{lm:||W||}, we then obtain
\begin{equation}
\big \| \chi [ U_\sigma , \i B^\sigma ] \chi \big \| \le \mathrm{C} \alpha^{\frac{1}{2}} \sigma.
\end{equation}
Hence \eqref{eq:[W_i,iB]}, $i=1$, is proven. In order to prove \eqref{eq:[W_i,iB]}, $i=2$, let us decompose
\begin{align}
[ W_2 , \i B^\sigma ] =& - [ \chi , \i B^\sigma ] W_\sigma \bar \chi \big [ H_{\bar \chi} - \lambda \big ]^{-1} \bar \chi W_\sigma \chi + \text{h.c.} \label{eq:[W_2,iB]_1} \\
& -\chi [ W_\sigma , \i B^\sigma ] \bar \chi \big [ H_{\bar \chi} - \lambda \big ]^{-1} \bar \chi W_\sigma \chi + \text{h.c.} \label{eq:[W_2,iB]_2} \\
& - \chi W_\sigma [ \bar \chi , \i B^\sigma ] \big [ H_{\bar \chi} - \lambda \big ]^{-1} \bar \chi W_\sigma \chi + \text{h.c.} \label{eq:[W_2,iB]_3} \\
& - \chi W_\sigma \bar \chi \big [ \big [ H_{\bar \chi} - \lambda \big ]^{-1} , \i B^\sigma \big ] \bar \chi W_\sigma \chi. \label{eq:[W_2,iB]_4}
\end{align}
Using Equations \eqref{eq:[Hf,iB]}, \eqref{eq:[chi,iB]}, \eqref{eq:[barchi,iB]} and \eqref{eq:[Phi,iB]} for the different commutators entering the terms \eqref{eq:[W_2,iB]_1}, \eqref{eq:[W_2,iB]_2} and \eqref{eq:[W_2,iB]_3}, one can check in the same way as in the proof of \eqref{eq:||W_i||}, $i=2$, in Lemma \ref{lm:||W||} that
\begin{align}\label{eq:[W_2,iB]_estimate1}
\big \| \eqref{eq:[W_2,iB]_1} + \eqref{eq:[W_2,iB]_2} + \eqref{eq:[W_2,iB]_3} \big \| \le \mathrm{C} \alpha^{\frac{1}{2}} \sigma.
\end{align}
To conclude we need to estimate \eqref{eq:[W_2,iB]_4}. We expand $[ H_{\bar \chi} - \lambda ]^{-1}$ into the Neumann series \eqref{eq:Neumann}, which leads to
\begin{align}
& \big [ \big [ H_{\bar \chi} - \lambda \big ]^{-1} , \i B^\sigma \big ] \notag \phantom{\sum_n}\\
&=- \big [ H_{\bar \chi} - \lambda \big ]^{-1} \big [ H_{\bar \chi} , \i B^\sigma \big ] \big [ H_{\bar \chi} - \lambda \big ]^{-1} \phantom{\sum_n} \notag \\
&= - \big [ H_{\bar \chi}^1 - \lambda \big ]^{-1} \sum_{n \ge 0} \Big ( - \bar \chi U_\sigma \bar \chi \big [ H_{\bar \chi}^1 - \lambda \big ]^{-1} \Big )^n \big [ H_{\bar \chi} , \i B^\sigma \big ] \notag \\
& \qquad \times \big [ H_{\bar \chi}^1 - \lambda \big ]^{-1} \sum_{n' \ge 0} \Big ( - \bar \chi U_\sigma \bar \chi \big [ H_{\bar \chi}^1 - \lambda \big ]^{-1} \Big )^{n'}. \label{eq:[W_2,iB]_6}
\end{align}
Inserting this series into \eqref{eq:[W_2,iB]_4} yields a sum of terms of the form
\begin{align}
& \chi W_\sigma \bar \chi \big [ H_{\bar \chi}^1 - \lambda \big ]^{-1} \Big ( \bar \chi U_\sigma \bar \chi \big [ H_{\bar \chi}^1 - \lambda \big ]^{-1} \Big )^n \big [ H_{\bar \chi} , \i B^\sigma \big ] \notag \\
& \times \big [ H_{\bar \chi}^1 - \lambda \big ]^{-1} \Big ( \bar \chi U_\sigma \bar \chi \big [ H_{\bar \chi}^1 - \lambda \big ]^{-1} \Big )^{n'} \bar \chi W_\sigma \chi, \label{eq:[W_2,iB]_6'}
\end{align}
where $n,n' \in \mathbb{N}$. To estimate \eqref{eq:[W_2,iB]_6'}, we notice that, by Lemma \ref{lm:standard2}, $W_\sigma \chi^\sigma_f = ( \mathds{1} \otimes \mathds{1}_{H_f \le 3\sigma} ) W_\sigma \chi^\sigma_f$, and likewise with $U_\sigma$ replacing $W_\sigma$. Thus, since $\mathds{1} \otimes H_f$ commutes with $H_{\bar \chi}^1$, we conclude from \eqref{eq:||W_2||_1} and \eqref{eq:Neumann_2} that
\begin{align}
\big \| \eqref{eq:[W_2,iB]_6'} \big \| \le \mathrm{C} \alpha^{\frac{1}{2}} \sigma \big ( \mathrm{C}' \alpha^{\frac{1}{2}} \big )^{n+n'} & \Big \| \big [ H_{\bar \chi}^1 - \lambda \big ]^{-\frac{1}{2}} ( \mathds{1} \otimes \mathds{1}_{H_f \le (2n+1) \sigma} ) \notag \\
& \big [ H_{\bar \chi} , \i B^\sigma \big ] ( \mathds{1} \otimes \mathds{1}_{H_f \le (2n'+1) \sigma} ) \big [ H_{\bar \chi}^1 - \lambda \big ]^{-\frac{1}{2}} \Big \|.
\end{align}
Using identities \eqref{eq:[Hf,iB]} and \eqref{eq:[chi,iB]}--\eqref{eq:[Phi,iB]}, one can check that, for any $\gamma>0$,
\begin{align}
\big \| \big [ H_{\bar \chi} , \i B^\sigma \big ] ( \mathds{1} \otimes \mathds{1}_{H_f \le \gamma \sigma} ) \big \| \le \mathrm{C} \gamma \sigma.
\end{align}
This implies
\begin{align}
\big \| \eqref{eq:[W_2,iB]_6'} \big \| \le \mathrm{C} \alpha^{\frac{1}{2}} \sigma (n+n'+1) \big ( \mathrm{C}' \alpha^{\frac{1}{2}} \big )^{n+n'}.
\end{align}
Summing over $n,n'$, we get that
\begin{align}\label{eq:[W_2,iB]_estimate3}
\big \| \eqref{eq:[W_2,iB]_4} \big \| \le \mathrm{C} \alpha^{\frac{1}{2}} \sigma,
\end{align}
for $\alpha$ small enough, which concludes the proof of \eqref{eq:[W_i,iB]}, $i=2$.
\end{proof}
In the proof of Theorem \ref{thm:Mourre}, it will be convenient to replace $F$ by an operator $\tilde F$, translated from $F$ in such a way that the unperturbed part in $\tilde F$ do not depend on the spectral parameter $\lambda$ anymore. More precisely, let
\begin{equation}
\tilde F := F + \lambda - E_\sigma.
\end{equation}
Then we have that $\tilde F = \tilde F_0 + W_1 + W_2$, where
\begin{equation}
\tilde F_0 := F_0 + \lambda - E_\sigma = \mathds{1} \otimes \big ( \frac{1}{2} P_f^2 + H_f \big ) - \nabla E_\sigma \otimes P_f,
\end{equation}
and $W_1$, $W_2$ are defined as in \eqref{eq:def_W1}, \eqref{eq:def_W2}.
\begin{lemma}\label{lm:Delta->Delta'}
For any $\mathrm{C}_0>0$, there exists $\alpha_0>0$ such that, for all $|P| \le p_c$, $0 \le \alpha \le \alpha_0$, $0 < \sigma \le \mathrm{C}_0 \alpha^{1/2}$, and $\lambda \in J_\sigma^<$,
\begin{equation}\label{eq:Delta->Delta'}
\mathds{1}_{\Delta_\sigma}(F) = \mathds{1}_{\Delta_\sigma}(F) \mathds{1}_{\Delta'_\sigma}( \tilde F ),
\end{equation}
where $\Delta'_\sigma := [ \rho \sigma / 16 , \rho \sigma / 8 ]$ and $\Delta_\sigma$ is given in \eqref{eq:Deltasigma_2}.
\end{lemma}
\begin{proof}
Since $\tilde F$ is a translate of $F$, it is only necessary to check that
$\Delta_\sigma\subseteq\Delta_\sigma'-\lambda+E_\sigma$ for all $\lambda\in J_\sigma^<$, or equivalently, that
$\Delta_\sigma\subseteq\Delta_\sigma'-J_\sigma^<+E_\sigma$ in the sense of ``sumsets''.
Using the definitions of $\Delta_\sigma$, $\Delta'_\sigma$, $J_\sigma^<$, and the fact that $|E - E_\sigma| \le \mathrm{C} \alpha \sigma$  by Proposition \ref{prop:Esigma(P)}, one can verify that this is the case for $\alpha$ sufficiently small.
\end{proof}
Let $f_\sigma \in \mathrm{C}_0^\infty(\mathbb{R} ; [0,1] )$ be such that $f_\sigma=1$ on $\Delta'_\sigma = [ \rho \sigma / 16 , \rho \sigma / 8 ]$ and
\begin{equation}
\mathrm{supp}(f_\sigma) \subset [ \frac{3}{64} \rho \sigma , \frac{9}{64} \rho \sigma ].
\end{equation}
\begin{lemma}\label{cor:f(F)-f(F_0)}
For any $\mathrm{C}_0>0$, there exists $\alpha_0>0$ such that, for all $|P| \le p_c$, $0 \le \alpha \le \alpha_0$, $0 < \sigma \le \mathrm{C}_0 \alpha^{1/2}$, and $\lambda \in J_\sigma^<$,
\begin{equation}
\big \| f_\sigma ( \tilde F ) - f_\sigma ( \tilde F_0 ) \big \| \le \mathrm{C} \alpha^{\frac{1}{2}}.
\end{equation}
\end{lemma}
\begin{proof}
Let $\tilde f_\sigma$ be an almost analytic extension of $f_\sigma$ obeying
\begin{align}
\mathrm{supp}( \tilde f_\sigma ) \subset \big \{ z \in \mathbb{C} , \Re(z) \in \mathrm{supp}(f_\sigma) , | \Im(z) | \le \sigma \big \}, \label{eq:tildef1}
\end{align}
$\partial_{\bar z} \tilde f_\sigma (z) = 0$ if $\Im(z)=0$, and
\begin{equation}\label{eq:tildef2}
\big | \frac{\partial \tilde f_\sigma}{ \partial \bar z } (z) \big | \le \frac{ \mathrm{C}_n }{ \sigma } \big ( \frac{ |y| }{ \sigma } \big )^n,
\end{equation}
for any $n\in\mathbb{N}$ (see for instance \cite{HS}). Here we used the notations
\begin{equation}
z = x + \i y, \quad \frac{ \partial }{ \partial \bar z} = \frac{ \partial }{ \partial x } + \i \frac{ \partial }{ \partial y }.
\end{equation}
By the Helffer-Sj{\" o}strand functional calculus and the second resolvent equation,
\begin{equation}\label{eq:f(F)-f(Ksigmachi)}
f_\sigma ( \tilde F ) - f_\sigma ( \tilde F_0 ) = \frac{ \i }{2\pi} \int \frac{ \partial \tilde f_\sigma }{ \partial \bar z }(z) \big [ \tilde F - z \big ]^{-1} \big ( \tilde F - \tilde F_0 \big ) \big [ \tilde F_0 - z \big ]^{-1} \d z \wedge \d \bar z.
\end{equation}
Lemma \ref{lm:||W||} implies
\begin{equation}\label{eq:F-Ksigmachi}
\big \| \tilde F - \tilde F_0 \big \| = \big \| F - F_0 \big \|  = \big \| W_1 + W_2 \| \le \mathrm{C} \alpha^{\frac{1}{2}} \sigma.
\end{equation}
The statement of the lemma then follows from \eqref{eq:tildef1}--\eqref{eq:F-Ksigmachi}.
\end{proof}
Lemma \ref{cor:f(F)-f(F_0)} will allow us to replace $f_\sigma( \tilde F)$ by $f_\sigma( \tilde F_0)$ in our proof of Theorem \ref{thm:Mourre}. In view of Lemma \ref{lm:[F0,iB]}, we shall also need to replace $f_\sigma(\tilde F_0)$ by some function of $H_f$. This is the purpose of the following lemma.
\begin{lemma}\label{lm:f(F_0)-f(Hf)}
For any $\mathrm{C}_0>0$, there exists $\alpha_0>0$ such that, for all $|P| \le p_c$, $0 \le \alpha \le \alpha_0$, $0 < \sigma \le \mathrm{C}_0 \alpha^{1/2}$, and $\lambda \in J_\sigma^<$,
\begin{equation}
f_\sigma( \tilde F_0 ) ( \mathds{1} \otimes \mathds{1}_{ \frac{1}{32} \rho \sigma \le H_f \le \frac{1}{4} \rho \sigma } ) = f_\sigma ( \tilde F_0 ).
\end{equation}
\end{lemma}
\begin{proof}
We recall that
\begin{equation}
	\tilde F_0=\tilde F_0(H_f,P_f)
	=  \mathds{1} \otimes \big ( \frac{1}{2} P_f^2 + H_f \big ) - \nabla E_\sigma \otimes P_f \,.
\end{equation}
The claim of the lemma is equivalent to the statement that whenever $\tilde F_0(X_0,X)\in  \mathrm{supp}( f_\sigma ) $
with $|X|\leq X_0$, then $X_0\in  [\frac{1}{32}\rho\sigma,\frac{1}{4}\rho\sigma] $.

Let $[a,b]\equiv [ \frac{3}{64} \rho \sigma , \frac{9}{64} \rho \sigma ]\supset \mathrm{supp}( f_\sigma )$. We assume that
\begin{equation}
	a \, \le \, \tilde F(X_0,X) \, = \, X_0+\frac12 X^2-\nabla E_\sigma\cdot X \, \le \, b
\end{equation}
with $|X|\leq X_0$. Clearly, this implies, on the one hand, that
\begin{equation}
	X_0-|\nabla E_\sigma|X_0 \, \le \,  \tilde F(X_0,X) \, \le \, b
\end{equation}
so that $X_0\le(1-|\nabla E_\sigma|)^{-1}b$, and, on the other hand,
\begin{equation}
	 X_0+\frac12 X_0^2 +|\nabla E_\sigma|X_0 \, \ge \, \tilde F(X_0,X) \, \ge \, a
\end{equation}
so that $X_0\ge(1+|\nabla E_\sigma|)^{-1}(a-\frac12(1-|\nabla E_\sigma|)^{-2}b^2)$.

By Proposition \ref{prop:Esigma(P)},
$| \nabla E_\sigma | \le |P| + \mathrm{C} \alpha \le 1/10$ for $|P| \le 1/40$
and $\alpha$ sufficiently small. Thus, one concludes that
$X_0\in [\frac{1}{32} \rho \sigma,\frac{1}{4} \rho \sigma]$, as claimed.
\end{proof}
We will also make use of the following easy lemma.
\begin{lemma}\label{lm:[F,iB]bounded}
For any $\mathrm{C}_0>0$, there exists $\alpha_0>0$ such that, for all $|P| \le p_c$, $0 \le \alpha \le \alpha_0$, $0 < \sigma \le \mathrm{C}_0 \alpha^{1/2}$, and $\lambda \in J_\sigma^<$, the operators $[ F , \i B^\sigma ] f_\sigma(\tilde F_0)$ and $[ F , \i B^\sigma ] f_\sigma(\tilde F)$ are bounded on $\mathrm{Ran}( P_\sigma \otimes \mathds{1} )$ and satisfy
\begin{equation}\label{eq:[F,iB]bounded}
\big \| [ F , \i B^\sigma ] f_\sigma(\tilde F_0) \big \| \le \mathrm{C} \sigma, \quad \big \| [ F , \i B^\sigma ] f_\sigma(\tilde F) \big \| \le \mathrm{C} \sigma.
\end{equation}
\end{lemma}
\begin{proof}
The first bound in \eqref{eq:[F,iB]bounded} is a consequence of Lemmata \ref{lm:[W,iB]} and \ref{lm:f(F_0)-f(Hf)}. Indeed, using expression \eqref{eq:[F0,iB]} for $[F_0,\i B^\sigma]$, we get
\begin{align}
\big \| [ F , \i B^\sigma ] f_\sigma(\tilde F_0) \big \| &\le \big \| [ F_0 , \i B^\sigma ] ( \mathds{1} \otimes \mathds{1}_{H_f \le \frac{1}{4} \rho \sigma } ) \big \| + \big \| [ W_1 , \i B^\sigma ] \big \| + \big \| [ W_2 , \i B^\sigma ] \big \| \notag \\
& \le \mathrm{C} \sigma.
\end{align}
Likewise, to prove the second bound in \eqref{eq:[F,iB]bounded}, it suffices to show that
\begin{equation}\label{eq:1(F)1(Hf)}
f_\sigma( \tilde F ) = ( \mathds{1} \otimes \mathds{1}_{H_f \le \rho \sigma } ) f_\sigma( \tilde F ).
\end{equation}
Since $\chi^\sigma_f \mathds{1}_{ H_f \le \rho \sigma } = \chi^\sigma_f$, and since $\tilde F_0$ commutes with $\mathds{1} \otimes \mathds{1}_{ H_f \le \rho \sigma }$, it follows that $\tilde F$ commutes with $\mathds{1} \otimes \mathds{1}_{H_f \le \rho \sigma}$. By Lemma \ref{lm:||W||},
\begin{align}
\tilde F ( \mathds{1} \otimes \mathds{1}_{H_f \ge \rho \sigma} ) \ge \tilde F_0 ( \mathds{1} \otimes \mathds{1}_{H_f \ge \rho \sigma} ) - \mathrm{C} \alpha^{\frac{1}{2}} \sigma ( \mathds{1} \otimes \mathds{1}_{H_f \ge \rho \sigma} ).
\end{align}
Using the fact that $| \nabla E_\sigma | \le 1/8$ for $|P|\le 1/40$ and $\alpha$ sufficiently small (see Proposition \ref{prop:Esigma(P)}), we obtain
\begin{align}
\tilde F_0 ( \mathds{1} \otimes \mathds{1}_{H_f \ge \rho \sigma} ) & = \big ( \mathds{1} \otimes \big ( \frac{1}{2} P_f^2 + H_f \big ) - \nabla E_\sigma \otimes P_f \big ) ( \mathds{1} \otimes \mathds{1}_{H_f \ge \rho \sigma} ) \notag \\
& \ge ( 1 - 2 | \nabla E_\sigma | ) ( \mathds{1} \otimes H_f ) ( \mathds{1} \otimes \mathds{1}_{H_f \ge \rho \sigma} ) \phantom{ \frac{1}{2} } \notag \\
& \ge \frac{3}{4} \rho \sigma ( \mathds{1} \otimes \mathds{1}_{H_f \ge 2 \rho \sigma} ) .
\end{align}
Hence, for $\alpha$ sufficiently small,
\begin{align}
\tilde F ( \mathds{1} \otimes \mathds{1}_{H_f \ge \rho \sigma} ) \ge \frac{ 1 }{ 2 } \rho \sigma ( \mathds{1} \otimes \mathds{1}_{H_f \ge \rho \sigma} ).
\end{align}
Since $\mathrm{supp}(f_\sigma) \subset [ 3 \rho \sigma / 64 , 9 \rho \sigma / 64 ]$, it follows that $(\mathds{1} \otimes \mathds{1}_{H_f \ge \rho \sigma} ) f_\sigma( \tilde F ) = 0$, which establishes \eqref{eq:1(F)1(Hf)} and concludes the proof.
\end{proof}
Next, we turn to the proof of Theorem \ref{thm:Mourre}. Recall that the intervals $\Delta_\sigma$, $\Delta'_\sigma$ are given by $\Delta_\sigma = [ - \rho \sigma / 128 , \rho \sigma / 128 ]$, $\Delta'_\sigma = [ \rho \sigma / 16 , \rho \sigma /8 ]$, and that the function $f_\sigma \in \mathrm{C}_0^\infty( \mathbb{R} ; [0,1] )$ is such that $f_\sigma = 1$ on $\Delta'_\sigma$ and $\mathrm{supp}(f_\sigma) \subset [3 \rho \sigma / 64 , 9 \rho \sigma / 64 ]$. Let us also recall the notations $\tilde F = F + \lambda - E_\sigma$, $\tilde F_0 = F_0 + \lambda - E_\sigma$. By Lemma \ref{lm:Delta->Delta'}, we have that
\begin{align}
& \mathds{1}_{\Delta_\sigma} (F) [ F , \i B^\sigma ] \mathds{1}_{\Delta_\sigma}(F) \notag \\
& = \mathds{1}_{\Delta_\sigma} (F) \mathds{1}_{\Delta'_\sigma} ( \tilde F ) [ F , \i B^\sigma ] \mathds{1}_{\Delta'_\sigma} ( \tilde F ) \mathds{1}_{\Delta_\sigma}(F) \label{eq:mourre0} \\
& = \mathds{1}_{\Delta_\sigma} (F) \mathds{1}_{\Delta'_\sigma} ( \tilde F ) f_\sigma ( \tilde F ) [ F , \i B^\sigma ] f_\sigma( \tilde F ) \mathds{1}_{\Delta'_\sigma} ( \tilde F ) \mathds{1}_{\Delta_\sigma}(F).
\end{align}
Next, we write
\begin{align}
&f_\sigma ( \tilde F ) [ F , \i B^\sigma ] f_\sigma( \tilde F ) \notag \\
& = f_\sigma( \tilde F_0 ) [ F , \i B^\sigma ] f_\sigma( \tilde F_0 )  \label{eq:mourre1} \\
& \quad + ( f_\sigma( \tilde F ) - f_\sigma( \tilde F_0 ) ) [ F , \i B^\sigma ] f_\sigma( \tilde F ) + f_\sigma( \tilde F_0 ) [ F , \i B^\sigma ] ( f_\sigma ( \tilde F ) - f_\sigma ( \tilde F_0 ) ). \label{eq:mourre2}
\end{align}
Lemmata \ref{cor:f(F)-f(F_0)} and \ref{lm:[F,iB]bounded} imply
\begin{equation}\label{eq:mourre4}
\| \eqref{eq:mourre2} \| \le \mathrm{C} \alpha^{\frac{1}{2}} \sigma.
\end{equation}
Using Lemmata \ref{lm:[F0,iB]}, \ref{lm:[W,iB]}, \ref{cor:f(F)-f(F_0)} and \ref{lm:f(F_0)-f(Hf)}, we estimate \eqref{eq:mourre1} from below as follows:
\begin{align}
&  f_\sigma( \tilde F_0 ) [ F , \i B^\sigma ] f_\sigma( \tilde F_0 ) \notag \\
& \ge  f_\sigma( \tilde F_0 ) [ F_0 , \i B^\sigma ] f_\sigma( \tilde F_0 )   - \mathrm{C} \alpha^{\frac{1}{2}} \sigma f_\sigma( \tilde F_0 )^2 \notag \\
& \ge f_\sigma( \tilde F_0 ) [F_0,\i B^\sigma] ( \mathds{1} \otimes \mathds{1}_{ \frac{1}{32} \rho \sigma \le H_f \le \frac{1}{4} \rho \sigma } ) f_\sigma( \tilde F_0 ) - \mathrm{C} \alpha^{\frac{1}{2}} \sigma f_\sigma( \tilde F_0 )^2 \notag \\
& \ge \frac{1}{2} f_\sigma( \tilde F_0 ) ( \mathds{1} \otimes H_f ) ( \mathds{1} \otimes \mathds{1}_{ \frac{1}{32}  \rho \sigma \le H_f \le \frac{1}{4} \rho \sigma } ) f_\sigma( \tilde F_0 ) - \mathrm{C}' \alpha^{\frac{1}{2}} \sigma f_\sigma( \tilde F_0 )^2 \notag \\
& \ge \frac{\rho \sigma }{64} f_\sigma( \tilde F_0 )^2 - \mathrm{C}' \alpha^{\frac{1}{2}} \sigma f_\sigma( \tilde F_0)^2 \notag \\
& \ge \frac{\rho \sigma}{64} f_\sigma( \tilde F )^2 - \mathrm{C}'' \alpha^{\frac{1}{2}} \sigma. \label{eq:[F,iB]ge_5}
\end{align}
Inequality \eqref{eq:[F,iB]ge_5} combined with \eqref{eq:mourre4} yield
\begin{align}
f_\sigma( \tilde F ) [ F , \i B^\sigma ] f_\sigma( \tilde F ) & \ge \frac{ \rho \sigma }{64} f_\sigma( \tilde F )^2 - \mathrm{C} \alpha^{\frac{1}{2}} \sigma \notag \\
&\ge \frac{ \rho \sigma }{128} f_\sigma( \tilde F )^2 - \mathrm{C} \alpha^{\frac{1}{2}} \sigma \big ( \mathds{1} - f_\sigma( \tilde F )^2 \big ), \label{eq:mourre6}
\end{align}
provided that $\alpha$ is sufficiently small. Multiplying both sides of \eqref{eq:mourre6} by $\mathds{1}_{\Delta'_\sigma}( \tilde F )$ gives
\begin{equation}
\mathds{1}_{\Delta'_\sigma}( \tilde F ) [ F , \i B^\sigma ] \mathds{1}_{\Delta'_\sigma}( \tilde F ) \ge \frac{ \rho \sigma }{128} \mathds{1}_{\Delta'_\sigma}( \tilde F ).
\end{equation}
Inserting this into \eqref{eq:mourre0} and using Lemma \ref{lm:Delta->Delta'} conclude the proof of the theorem. \hfill $\square$

\appendix

\section{Technical estimates}\label{appendix:estimates}

In this appendix we collect some estimates that were used in Sections \ref{section:Feshbach} and \ref{section:mourre}. For $f: \mathbb{R}^3 \times \mathbb{Z}_2 \mapsto \mathbb{C}$ and $\gamma > 0$, we define
\begin{equation}
f^\gamma( k,\lambda ) = f( k,\lambda ) \mathds{1}_{ |k| \le \gamma }.
\end{equation}
Similarly we set
\begin{equation}
H_f^\gamma = \sum_{\lambda=1,2} \int_{ |k| \le \gamma } |k| a^*_\lambda( k ) a_\lambda( k ) \d k.
\end{equation}
We begin with two well-known lemmata; (see for instance \cite{BFS} for a proof).
\begin{lemma}\label{lm:standard1}
For any $f \in \mathrm{L}^2( \mathbb{R}^3 \times \mathbb{Z}_2 )$ such that $|k|^{-1/2} f \in \mathrm{L}^2( \mathbb{R}^3 \times \mathbb{Z}_2 )$, and any $\gamma > 0$,
\begin{align}
& \| a( f^\gamma ) [ H_f^\gamma + \gamma ]^{-1/2} \| \le \| |k|^{-\frac{1}{2}} f^\gamma \|, \\
& \| a^*( f^\gamma ) [ H_f^\gamma + \gamma ]^{-1/2} \| \le \| |k|^{-\frac{1}{2}} f^\gamma \| + \gamma^{-\frac{1}{2}} \| f^\gamma \|.
\end{align}
\end{lemma}
\begin{lemma}\label{lm:standard2}
For any $f \in \mathrm{L}^2( \mathbb{R}^3 \times \mathbb{Z}_2 )$, and any $\gamma > 0$, $\delta>0$,
\begin{align}
& a(f^\gamma) \mathds{1}_{ H_f^\gamma \le \delta } = \mathds{1}_{ H_f^\gamma \le \delta } a(f^\gamma) \mathds{1}_{ H_f^\gamma \le \delta } \\
& a^*(f^\gamma) \mathds{1}_{ H_f^\gamma \le \delta } = \mathds{1}_{ H_f^\gamma \le \gamma + \delta } a^*(f^\gamma) \mathds{1}_{ H_f^\gamma \le \delta }
\end{align}
\end{lemma}
\begin{proof}
The statement of the lemma follows directly from the ``pull-through formula''
\begin{equation}
a(k) g( H_f^\gamma ) = g( H_f^\gamma + |k| ) a(k),
\end{equation}
which holds for any bounded measurable function $g : [ 0 , \infty ) \to \mathbb{C}$, and any $k \in \mathbb{R}^3$, $|k| \le \gamma$.
\end{proof}
In the following, the parameters $\alpha$, $\sigma$ and $P$ are fixed with $ 0 \le \alpha \le \alpha_0$, where $\alpha_0$ is sufficiently small, $0 < \sigma \le \mathrm{C}_0 \alpha^{1/2}$, where $\mathrm{C}_0$ is a positive constant, and $|P| \le p_c = 1/40$. We use the notations introduced in Section \ref{section:decomp}.
\begin{lemma}\label{lm:appendix0}
For any $\mathrm{c} \ge 1/2$, we have that
\begin{equation}\label{eq:lm_Hsigmac}
K_\sigma \otimes \mathds{1} + \mathds{1} \otimes \big ( \frac{1}{2} P_f^2 + \mathrm{c} H_f \big ) - \nabla K_\sigma \otimes P_f \ge E_\sigma.
\end{equation}
In particular,
\begin{equation}\label{eq:Hf_le_Hsigma}
\mathds{1} \otimes H_f \le 2 ( H_\sigma -  E_\sigma ).
\end{equation}
\end{lemma}
\begin{proof}
To simplify notations, we set
\begin{equation}\label{eq:Hsigmac}
H_{\sigma,\mathrm{c}} = H_{\sigma,\mathrm{c}}(P) = K_\sigma \otimes \mathds{1} + \mathds{1} \otimes \big ( \frac{1}{2} P_f^2 + \mathrm{c} H_f \big )- \nabla K_\sigma \otimes P_f.
\end{equation}
Note that
\begin{align}
H_{\sigma,\mathrm{c}} &= \frac{1}{2} \big ( P - P_f - \alpha^{\frac{1}{2}} A_\sigma \big ) ^2 + H_f \otimes \mathds{1} + \mathrm{c} \, \mathds{1} \otimes H_f \notag \\
& = \frac{1}{2} \big ( \nabla H_\sigma \big )^2 + H_f \otimes \mathds{1} + \mathrm{c} \, \mathds{1} \otimes H_f. \label{eq:Hsigmac_1}
\end{align}
Let $\Phi \in D( H_{\sigma,\mathrm{c}})$, $\| \Phi \|=1$. We propose to show that
\begin{equation}\label{eq:lm_Feshbach1}
( \Phi , H_{\sigma,\mathrm{c}} \Phi ) \ge E_\sigma.
\end{equation}
Since the number operator $N^\sigma = \sum_{\lambda=1,2} \int_{|k|\le\sigma} a^*_\lambda(k) a_\lambda( k) \d k$ commutes with $H_{\sigma,\mathrm{c}}$, in order to prove \eqref{eq:lm_Feshbach1}, it suffices to consider $\Phi \in D( H_{\sigma,\mathrm{c}} )$ of the form $\Phi = \Phi_1 \otimes \Phi_2$ where $\Phi_1 \in \mathcal F_\sigma $ and $\Phi_2$ is an eigenstate of $N^\sigma |_{\mathcal{F}^\sigma}$. Let us prove the following assertion by induction:
\begin{itemize}
\item[$\mathbf{(h_n)}$] For all $\Phi = \Phi_1 \otimes \Phi_2 \in D( H_{\sigma,\mathrm{c}} )$ such that $\| \Phi_1 \| = \| \Phi_2 \| = 1$ and $N^\sigma \Phi_2 = n \Phi_2$,  \eqref{eq:lm_Feshbach1} holds.
\end{itemize}
Since $H_{\sigma,\mathrm{c}} ( \Phi_1 \otimes \Omega ) = ( K_\sigma \Phi_1 ) \otimes \Omega$ and since $E_\sigma = \inf \sigma ( K_\sigma )$ (see Proposition \ref{prop:Esigma(P)}), $\mathbf{(h_0)}$ is obviously satisfied.
Assume that $\mathbf{(h_n)}$ holds and let $\Phi = \Phi_1 \otimes \Phi_2 \in D( H_{\sigma,c} )$ with $\| \Phi_1 \| = \| \Phi_2 \| = 1$ and $N^\sigma \Phi_2 = (n+1) \Phi_2$. Let us write
\begin{equation}
\Phi_2 \big ( (k,\lambda),(k_1,\lambda_1),\dots,(k_n,\lambda_n) \big ) = \Phi_2 ( k,\lambda )\big ( (k_1,\lambda_1),\dots,(k_n,\lambda_n) \big ).
\end{equation}
One can compute
\begin{align}
( \Phi , H_{\sigma,\mathrm{c}} \Phi ) = \sum_{\lambda=1,2} \int_{ |k| \le \sigma } \big ( & \Phi_1 \otimes \Phi_2 ( k,\lambda ) , \notag \\
&( H_{\sigma,\mathrm{c}} ( P - k ) + \mathrm{c} |k| )  \Phi_1 \otimes \Phi_2 ( k,\lambda ) \big ) \d k.
\end{align}
Next, it follows from \eqref{eq:Hsigmac_1} that
\begin{align}
& \big ( \Phi_1 \otimes \Phi_2 ( k,\lambda ) , ( H_{\sigma,\mathrm{c}} ( P - k ) + \mathrm{c} |k| ) \Phi_1 \otimes \Phi_2 ( k,\lambda ) \big ) \notag \\
& = \Big ( \Phi_1 \otimes \Phi_2 ( k,\lambda ) , \Big ( H_{\sigma,\mathrm{c}} - k \cdot \nabla H_\sigma + \frac{k^2}{2} + \mathrm{c} |k| \Big ) \Phi_1 \otimes \Phi_2 ( k,\lambda ) \Big ).
\end{align}
Using that $k\cdot \nabla H_\sigma \le |k| / 4 + |k| ( \nabla H_\sigma )^2$ and that $ ( \nabla H_\sigma )^2 \le 2 H_{\sigma,\mathrm{c}}$, we obtain that
\begin{align}
& \big ( \Phi_1 \otimes \Phi_2 ( k,\lambda ) , ( H_{\sigma,\mathrm{c}} ( P - k ) + \mathrm{c} |k| ) \Phi_1 \otimes \Phi_2 ( k,\lambda ) \big ) \notag \\
& \ge \Big ( \Phi_1 \otimes \Phi_2 ( k,\lambda ) , \Big ( H_{\sigma,\mathrm{c}} - |k| ( \nabla H_\sigma )^2 + \frac{k^2}{2} + ( \mathrm{c} - \frac{1}{4} ) |k| \Big ) \Phi_1 \otimes \Phi_2 ( k,\lambda ) \Big ) \notag \\
& \ge \Big ( \Phi_1 \otimes \Phi_2 ( k,\lambda ) , \Big ( ( 1 - 2 |k| ) H_{\sigma,\mathrm{c}} + ( \mathrm{c} - \frac{1}{4} ) |k| \Big ) \Phi_1 \otimes \Phi_2 ( k,\lambda ) \Big ).
\end{align}
Since by the induction hypothesis $( \Phi_1 \otimes \Phi_2(k,\lambda) , H_{\sigma,\mathrm{c}} \Phi_1 \otimes \Phi_2(k,\lambda) ) \ge E_\sigma \| \Phi_2(k,\lambda) \|^2 $, this implies
\begin{align}
& \big ( \Phi_1 \otimes \Phi_2 ( k,\lambda ) , ( H_{\sigma,\mathrm{c}}( P - k ) + |k| ) \Phi_1 \otimes \Phi_2 ( k,\lambda ) \big ) \notag \\
& \ge \Big ( ( 1 - 2 |k| ) E_\sigma + (\mathrm{c} - \frac{ 1}{4}) |k| \Big ) \| \Phi_2 ( k,\lambda ) \|^2 \notag \\
& \ge \Big ( E_\sigma + |k| ( \mathrm{c} - \frac{1}{4} - 2 E_\sigma ) \Big ) \| \Phi_2(k,\lambda) \|^2.
\end{align}
By
Rayleigh-Ritz (see Proposition \ref{prop:Esigma(P)}),
\begin{equation}\label{eq:Esigma_small}
E_\sigma \le \frac{1}{2} P^2 + \mathrm{C} \alpha
\le \frac{1}{100}
\end{equation}
for $\alpha$ sufficiently small and $|P| \le 1/40$, so that, in particular, $\mathrm{c} - 1/4 - 2 E_\sigma \ge 0$; (recall that $\mathrm{c} \ge 1/2$). Therefore
$\mathbf{(h_{n+1})}$ holds, and hence \eqref{eq:lm_Feshbach1} is proven.

To prove \eqref{eq:Hf_le_Hsigma}, it suffices to write, using \eqref{eq:lm_Hsigmac} with $\mathrm{c}=1/2$,
\begin{align}
H_\sigma &= K_\sigma \otimes \mathds{1} + \mathds{1} \otimes \big ( \frac{1}{2} P_f^2 + \frac{1}{2}  H_f \big ) - \nabla K_\sigma \otimes P_f + \frac{1}{2} \big ( \mathds{1} \otimes H_f \big ) \notag \\
& \ge E_\sigma + \frac{1}{2} \big ( \mathds{1} \otimes H_f \big ).
\end{align}
\end{proof}
\begin{lemma}\label{lm:appendix1}
Let $0 < \delta < 1$. Then
\begin{equation}\label{eq:lm_appendix1}
H_\sigma ( \mathds{1} \otimes \mathds{1}_{H_f \le \delta} ) \ge ( 1 - \delta ) \big ( K_\sigma \otimes \mathds{1} \big ) ( \mathds{1} \otimes \mathds{1}_{H_f \le \delta} ).
\end{equation}
\end{lemma}
\begin{proof}
Note that $\mathds{1} \otimes \mathds{1}_{H_f \le \delta }$ commutes both with $H_\sigma$ and $K_\sigma \otimes \mathds{1}$. In addition, since the number operator $N^\sigma$ also commutes with $H_\sigma$ and $K_\sigma \otimes \mathds{1}$, it suffices to prove \eqref{eq:lm_appendix1} on states $\Phi \in D(H_\sigma)$ of the form $\Phi = \Phi_1 \otimes \Phi_2$ with $\| \Phi_1 \| = \| \Phi_2 \| = 1$, $\Phi_1 \in D(K_\sigma)$, and $\Phi_2 \in \mathrm{Ran}( \mathds{1}_{H_f \le \delta})$ is an eigenstate of $N^\sigma |_{\mathcal{F}^\sigma}$. For such a vector $\Phi$, we have
\begin{align}
\big ( \Phi , H_\sigma \Phi \big ) &= \big ( \Phi_1 , K_\sigma \Phi_1 \big ) + \big ( \Phi_2 , ( \frac{1}{2} P_f^2 + H_f ) \Phi_2 \big ) \notag \\
&\quad - \big ( \Phi_1 , \nabla K_\sigma \Phi_1 \big ) \big ( \Phi_2 , P_f \Phi_2 \big ). \label{eq:lm_appendix1_2}
\end{align}
One can check that
\begin{align}
&\big | \big ( \Phi_1 , \nabla K_\sigma \Phi_1 \big ) \big | \le \big ( \Phi_1 , ( \nabla K_\sigma )^2 \Phi_1 \big )^{1/2}, \\
&\big | \big ( \Phi_2 , P_f \Phi_2 \big ) \big | \le \big ( \Phi_2 , H_f \Phi_2 \big ),
\end{align}
and hence
\begin{align}
\big ( \Phi_1 , \nabla K_\sigma \Phi_1 \big ) \big ( \Phi_2 , P_f \Phi_2 \big ) \le \frac{1}{2} \big ( \Phi_1 , ( \nabla K_\sigma )^2 \Phi_1 \big ) \big ( \Phi_2 , H_f \Phi_2 \big ) + \frac{1}{2} \big ( \Phi_2 , H_f \Phi_2 \big ).
\end{align}
Inserting this into \eqref{eq:lm_appendix1_2} and using that $( \nabla K_\sigma )^2 \le 2 K_\sigma$, we obtain
\begin{align}
\big ( \Phi , H_\sigma \Phi \big ) &\ge \big ( \Phi_1 , K_\sigma \Phi_1 \big ) + \frac{1}{2} \big ( \Phi_2 , H_f \Phi_2 \big ) - \frac{1}{2} \big ( \Phi_1 , (\nabla K_\sigma )^2 \Phi_1 \big ) \big ( \Phi_2 , H_f \Phi_2 \big ) \notag \\
&\ge \big ( \Phi_1 , K_\sigma \Phi_1 \big ) - \delta \big ( \Phi_1 , \frac{1}{2} (\nabla K_\sigma )^2 \Phi_1 \big ) \notag \\
&\ge ( 1 - \delta ) \big ( \Phi_1 , K_\sigma \Phi_1 \big ), \phantom{ \frac{1}{2} }
\end{align}
which concludes the proof.
\end{proof}
%
%

\section{The smooth Feshbach-Schur map}\label{appendix:Feshbach}

In this appendix we recall the definition and some of the main properties of the smooth Feshbach-Schur map introduced in \cite{BCFS}. The version we present uses aspects developed in \cite{GH} and \cite{FGS3}.

Let $\mathcal{H}$ be a separable Hilbert space. Let $\chi$,  $\bar\chi$ be nonzero bounded operators on $\mathcal{H}$, such that $[ \chi , \bar \chi ] = 0$ and $\chi^{2}+ \bar\chi^{2} = 1$. Let $H$ and $T$ be two closed operators on $\mathcal{H}$ such that $D(H)=D(T)$. Define $W = H - T$ on $D(T)$ and
\begin{align}
H_\chi = T + \chi W \chi, \quad H_{\bar\chi} = T + \bar\chi W \bar\chi
\end{align}
We make the following hypotheses:
\begin{itemize}
\item[(1)] $\chi T \subset T \chi$ and $\bar\chi T \subset T \bar\chi$.
\item[(2)] $T,H_{\bar\chi}: D(T) \cap \mathrm{Ran}(\bar\chi) \to \mathrm{Ran}(\bar\chi)$ are bijections with bounded inverses.
\item[(3)] $W \chi$ and $\chi W$ extend to bounded operators on $\mathcal{H}$.
\end{itemize}
Given the above assumptions, the (smooth) Feshbach-Schur map $F_\chi(H)$ is defined by
\begin{equation}
F_\chi(H) = H_\chi - \chi W \bar\chi H_{\bar\chi}^{-1} \bar\chi W \chi.
\end{equation}
Note that $F_\chi(H)$ is well-defined on $D(T)$. If Hypotheses (1),(2),(3) above are satisfied, we say that $H$ is in the domain of $F_\chi$. In addition, we consider the two auxiliary bounded operators $Q_\chi(H)$ and $Q_\chi^{\#}(H)$ defined by
\begin{equation}
Q_\chi(H) = \chi - \bar\chi H_{\bar\chi}^{-1} \bar\chi W \chi, \quad Q_\chi^{\#}(H) = \chi - \chi W \bar\chi H_{\bar\chi}^{-1} \bar\chi.
\end{equation}
It follows from \cite{BCFS,GH,FGS3} that the smooth Feshbach-Schur map $F_\chi$ is isospectral in the following sense:
\begin{theorem}\label{thm:prop_Feshbach}
Let $H,T,\chi,\bar\chi$ be as above. Then the following holds:
\begin{itemize}
\item[(i)] Let $V$ be a subspace such that $\mathrm{Ran} \chi \subset V \subset \mathcal{H}$, $T:D(T)\cap V \to V$ and $\bar\chi T^{-1} \bar\chi V \subset V$. Then $H : D(T) \to \mathcal{H}$ is bounded invertible if and only if $F_\chi(H) : D(T) \cap V \to V$ is bounded invertible, and we have
\begin{align}
&H^{-1} = Q_\chi(H) F_\chi(H)^{-1} Q_\chi^{\#}(H) + \bar\chi H_{\bar\chi}^{-1} \bar\chi, \label{eq:H^-1=}\\
&F_\chi(H)^{-1} = \chi H^{-1} \chi + \bar \chi T^{-1} \bar\chi. \label{eq:F^-1=}
\end{align}
\item[(ii)] If $\phi \in \mathcal{H} \setminus \{0\}$ solves $H \phi = 0$ then $\psi := \chi \phi \in \mathrm{Ran} \chi \setminus \{0\}$ solves $F_{\chi} (H) \, \psi = 0$.
\item[(iii)] If $\psi \in \mathrm{Ran}\, \chi \setminus \{0\}$ solves $F_{\chi} (H) \, \psi = 0$ then $\phi := Q_{\chi}(H) \psi \in \mathcal{H} \setminus \{0\}$ solves $H \phi = 0$.
\item[(iv)] The multiplicity of the spectral value $\{0\}$ is conserved in the sense that 
\begin{equation}
\dim \mathrm{Ker} H = \dim \mathrm{Ker} F_{\chi} (H).
\end{equation}
\end{itemize}
\end{theorem}

Next, we recall a result given in \cite{FGS3} showing that a LAP for $H$ can be deduced from a corresponding LAP for $F_\chi(H-\lambda)$, for suitably chosen $\lambda$'s. Notice that, in \cite{FGS3}, $F_\chi( H - \lambda )$ is considered as an operator on $\mathcal{H}$, whereas its restriction to some closed subspace $V$ is considered here. However, the the following theorem can be proven is the same way. For the convenience of the reader, we recall the proof.

\begin{theorem} \label{thm:LAPtransfer}
Let $H,T,\chi,\bar\chi$ be as above. Let $\Delta$ be an open interval in $\mathbb{R}$. Let $V$ be a closed subspace of $\mathcal{H}$ satisfying the assumptions of Theorem \ref{thm:prop_Feshbach}$(i)$. Let $B$ a self-adjoint operator on $\mathcal{H}$ such that $B:D(B) \cap V \to V$ and $[B \pm \i]^{-1} V \subset V$. Assume that $\forall \lambda \in \Delta$,
\begin{equation}
[ A_\lambda , B ] \ \mbox{extends to a bounded operator}, \label{eqn:5}
\end{equation}
where $A_\lambda$ stands for one of the operators $A_\lambda = \chi,\ \overline{\chi},\ \chi W,\ W\chi,\  \bar\chi [ H_{\bar\chi} - \lambda ]^{-1} \bar\chi$.
If $H - \lambda$ is in the domain of $F_{\chi}$, then for any $\nu\ge 0$ and $0 < s \le 1$,
\begin{align}
& \lambda \mapsto \langle B \rangle^{-s} (F_\chi ( H - \lambda ) - \i 0 )^{-1} \langle B \rangle^{-s} \in C^\nu(\Delta ; \mathcal{B}( V ) ) \notag \\
& \text{implies that}\quad \lambda \mapsto \langle B \rangle^{-s} ( H - \lambda - \i 0 )^{-1} \langle B \rangle^{-s} \in C^\nu(\Delta ; \mathcal{B} ( \mathcal{H} ) ). \label{eqn:7}
\end{align}
\end{theorem}
\begin{proof}
It follows form Equation \eqref{eq:H^-1=} with $H$ replaced by $H - \lambda - \i \varepsilon$ that
\begin{align}
[ H - \lambda - \i \varepsilon ]^{-1} =& Q_\chi(H-\lambda-\i\varepsilon) F_\chi(H-\lambda-\i\varepsilon)^{-1} Q_\chi^{\#}(H-\lambda-\i\varepsilon) \notag \\
& + \bar \chi [ H_{\bar\chi} - \lambda -\i \varepsilon ]^{-1} \bar\chi.
\end{align}
The map $\varepsilon \mapsto [ H_{\bar\chi} - \lambda - \i \varepsilon ]^{-1} \in \mathcal{B}( \mathrm{Ran}( \bar \chi ))$ is analytic in a neighborhood of 0, and can be expanded as
\begin{equation}
[ H_{\bar\chi} - \lambda - \i \varepsilon ]^{-1} = [ H_{\bar\chi} - \lambda ]^{-1} + \i \varepsilon [ H_{\bar\chi} - \lambda ]^{-1} \bar\chi^2 [ H_{\bar\chi} - \lambda]^{-1} + O(\varepsilon^2).
\end{equation}
This yields
\begin{equation} \label{BFB}
\lim_{ \varepsilon \rightarrow 0} \langle B \rangle^{-s} F_\chi ( H - \lambda - \i \varepsilon)^{-1} \langle B \rangle^{-s} = \langle B \rangle^{-s} [ F_\chi ( H - \lambda ) - \i 0 ]^{-1} \langle B \rangle^{-s}.
\end{equation}
Note that
\begin{equation}
\langle B \rangle^{-s} = \mathrm{C}_s \int_0^\infty \frac{\d\omega}{\omega^{s/2}}(\omega+1+B^2)^{-1},
\end{equation}
where $\mathrm{C}_s:=\big[\int_0^\infty \frac{\d\omega}{\omega^{s/2}}(\omega+1)^{-1}\big]^{-1}$.
Hence, Conditions \eqref{eqn:5} imply that the operators
\begin{equation}
\langle B \rangle^{-s} \chi \langle B \rangle^{s},\quad \langle B \rangle^{-s} \overline{\chi} \langle B \rangle^{s},\quad \langle B \rangle^{-s} \chi \langle B \rangle^s,\quad \langle B \rangle^{-s} \bar \chi \langle B \rangle^s
\end{equation}
are bounded. Similarly, the maps
\begin{equation}
\lambda \mapsto \langle B \rangle^{-s} \bar \chi [ H_{\bar\chi} - \lambda ]^{-1} \bar\chi \langle B \rangle^{s} \quad \text{ and } \quad  \lambda \mapsto \langle B \rangle^{s} \bar \chi [ H_{\bar\chi} - \lambda ]^{-1} \bar\chi \langle B \rangle^{-s}
\end{equation}
are in $C^\infty(\Delta ; \mathcal{B}( \mathcal{H} ))$. This property shows that
\begin{equation}
\langle B \rangle^s Q_\chi(H-\lambda) \langle B \rangle^{-s} \quad \text{ and } \quad \langle B \rangle^s Q^\#_\chi( H - \lambda ) \langle B \rangle^{-s}
\end{equation}
are bounded and smooth in $\lambda \in \Delta$. The theorem then follows from \eqref{BFB}, the fact that $H - \lambda$ is in the domain of $F_\chi$, and \eqref{eq:H^-1=}.
\end{proof}

\section{Bound particles coupled to a quantized radiation field}\label{appendix:bound}

In this appendix, we explain how to adapt the proof of Theorem \ref{thm:LAP} to the case of non-relativistic particles interacting with an infinitely heavy nucleus and coupled to a massless radiation field. To simplify matters, we assume that the non-relativistic particles are spinless, and that the bosons are scalar (Nelson's model). The Hamiltonian $H^{\mathrm{N}}$ associated to this system acts on $\mathcal{H} = \mathcal{H}_{\mathrm{el}} \otimes \mathcal{F}$, where $\mathcal{H}_{\mathrm{el}} = \mathrm{L}^2( \mathbb{R}^{3N} )$, and $\mathcal{F} = \Gamma_s( \mathrm{L}^2( \mathbb{R}^3 ) )$ is the symmetric Fock space over $\mathrm{L}^2(\mathbb{R}^3)$. It is given by
\begin{equation}\label{eq:def_HgN}
H^{\mathrm{N}} := H_{\mathrm{el}} \otimes \mathds{1} + \mathds{1} \otimes H_f + W .
\end{equation}
Here, $H_{\mathrm{el}} = \sum_{j=1}^N p_j^2/2m_j + V$ denotes an $N$-particle Schr\"odinger operator on $\mathcal{H}_{\mathrm{el}}$.  For $k$ in $\mathbb{R}^3$, we denote by $a^*(k)$ and $a(k)$ the usual phonon creation and annihilation operators on $\mathcal{F}$ obeying the canonical commutation relations
\begin{equation} \label{commrelations}
\left [ a^*(k) , a^*(k') \right ] = \left [ a(k) , a(k') \right ] = 0 \quad,\quad \left [ a(k) , a^*(k') \right ] = \delta (k-k').
\end{equation}
The operator associated with the energy of the free boson field, $H_f$, is given by the expression \eqref{eq:def_Hf}, except that the operators $a^*(k)$ and $a(k)$ now are scalar creation and annihilation operators as given above. The interaction $W$ in \eqref{eq:def_HgN} is assumed to be of the form $W = g \phi ( G_x )$ where $g$ is a small coupling constant, $x=(x_1,x_2,\dots,x_n)$ and
\begin{equation}
\label{eq:FormFactor}
\phi ( G_x ) := \frac{1}{\sqrt{2}} \sum_{j=1}^N \int_{ \mathbb{R}^3 } \frac{ \kappa^\Lambda(k) }{ |k|^{1/2-\mu} } \left [ e^{-ik \cdot x_j} a^*(k) + e^{ik \cdot x_j} a(k) \right ] \d k.
\end{equation}
As above, the function $\kappa^\Lambda$ denotes an ultraviolet cutoff, and the parameter $\mu$ is assumed to be non-negative.

We assume that $V$ is infinitely small with respect to $\sum_j p_j^2$, and that the spectrum of $H_{\mathrm{el}}$ consists of a sequence of discrete eigenvalues, $e_0, e_1 ,\dots$, below some semi-axis $[ \Sigma , \infty )$. Let $E^{\mathrm{N}} := \inf ( \sigma ( H^{\mathrm{N}} ) )$ and $y := \i \nabla_k$. Adapting the proof of Theorem \ref{thm:LAP}, one can show the following
\begin{theorem}\label{thm:Nelson}
Let $H^{\mathrm{N}}$ be given as above. For any $\mu \ge 0$, there exists $g_0>0$ such that, for any $0 \le g \le g_0$, $1/2 < s \le 1$, and any compact interval $J \subset (E^{\mathrm{N}}, ( e_0 + e_1 )/2 )$,
\begin{equation}\label{eq:LAP_Nelson}
\sup_{ z \in J_\pm } \big \| (  \d \Gamma ( \langle y \rangle ) + 1 )^{- s} \big [ H^{\mathrm{N}} - z \big ]^{-1} ( \d \Gamma ( \langle y \rangle )  + 1 )^{-s} \big \| \le \mathrm{C},
\end{equation}
where $\mathrm{C}$ is a positive constant depending on $J$ and $s$. In particular, the spectrum of $H^{\mathrm{N}}$ in $( E^{\mathrm{N}} , (e_0+e_1)/2 )$ is absolutely continuous.
Moreover, the map
\begin{equation}
J \ni \lambda \mapsto ( \d \Gamma ( \langle y \rangle ) + 1 )^{- s} \big [ H^{\mathrm{N}} - \lambda \pm \i 0 \big ]^{-1} ( \d \Gamma ( \langle y \rangle ) + 1 )^{- s} \in B( \mathcal{H} )
\end{equation}
is uniformly H\"older continuous in $\lambda$ of order $s-1/2$.
\end{theorem}
Let us emphasize that Theorem \ref{thm:Nelson} does not require any infrared regularization in the form factor. In comparison, the proof of \cite{FGS1} would give Theorem \ref{thm:Nelson} for any $\mu \ge 1$, and the one in \cite{FGS3} for any $\mu>0$. However, for the standard model of non-relativistic QED (which is considered in \cite{FGS1} and \cite{FGS3}), thanks to a Pauli-Fierz transformation, the methods given in \cite{FGS1} and \cite{FGS3} work without any infrared regularization.
\begin{proof}
We briefly explain how to adapt the proof of Theorem \ref{thm:LAP}. First, using the generator of dilatations on Fock space, $B$, as a conjugate operator, it follows from standard estimates that a Mourre estimate holds outside a neighborhood of $E_g^{\mathrm{N}}$; see \cite{BFS}.

To obtain the LAP near $E^{\mathrm{N}}$, we modify Sections \ref{section:Feshbach} and \ref{section:mourre} as follows: We take $T_\sigma = H_{\sigma}^{\mathrm{N}}$, where $H_{\sigma}^{\mathrm{N}}$ is the infrared cutoff Hamiltonian
\begin{equation}
H_{\sigma}^{\mathrm{N}} := H_{\mathrm{el}} \otimes \mathds{1} + \mathds{1} \otimes H_f + W_{\sigma} .
\end{equation}
Here $W_{\sigma} = g \phi( G_{x,\sigma} )$, and $\phi( G_{x,\sigma} )$ is given by \eqref{eq:FormFactor} except that the integral over $\mathbb{R}^3$ is replaced by the integral over $\{ k \in \mathbb{R}^3 , |k| \ge \sigma \}$. We define similarly $W^\sigma = H^{\mathrm{N}} - H_{\sigma}^{\mathrm{N}} = g \phi( G_{x}^\sigma )$ with the obvious notation. The Hilbert space $\mathcal{H}$ is unitarily equivalent to $\mathcal{H}_\sigma \otimes \mathcal{F}^\sigma$, where $\mathcal{H}_\sigma = \mathcal{H}_{\mathrm{el}} \otimes \mathcal{F}_\sigma$ and $\mathcal{F}_\sigma = \Gamma_s ( \mathrm{L}^2 ( \{ k \in \mathbb{R}^3 , |k| \ge \sigma \} ) )$, respectively $\mathcal{F}^\sigma = \Gamma_s ( \mathrm{L}^2 ( \{ k \in \mathbb{R}^3 , |k| \le \sigma \} ) )$. In this representation, we can write
\begin{equation}
H^{\mathrm{N}} = K_{\sigma}^{\mathrm{N}} \otimes \mathds{1} + \mathds{1} \otimes H_f + W^\sigma,
\end{equation}
where $K_{\sigma}^{\mathrm{N}}$ denotes the restriction of $H_{\sigma}^{\mathrm{N}}$ to $\mathcal{H}_\sigma$. It is known that the ground state energy $E_{\sigma}^{\mathrm{N}}$ of $K_{\sigma}^{\mathrm{N}}$ is separated from the rest of the spectrum by a gap of order $O(\sigma)$. Thus, letting $P_\sigma = \mathds{1}_{ \{ E_{\sigma}^{\mathrm{N}} \} }( K_{\sigma}^{\mathrm{N}} )$ and $\chi = P_\sigma \otimes \chi^\sigma_f$, one can define the smooth Feshbach-Schur operator in the same way as in Section \ref{section:Feshbach}, that is
\begin{align}
F(\lambda) &= F_\chi( H^{\mathrm{N}} - \lambda ) |_{\mathrm{Ran}(P_\sigma \otimes \mathds{1} ) } \notag \\
&= E_{\sigma}^{\mathrm{N}} - \lambda + \mathds{1} \otimes H_f + \chi W^\sigma \chi - \chi W^\sigma \bar\chi \big [ H_{\bar\chi} - \lambda \big ]^{-1} \bar \chi W^\sigma \chi,
\end{align}
for $\lambda$ in a neighborhood of $E_{\sigma}^{\mathrm{N}}$. The proof of the Mourre estimate for $F(\lambda)$ follows then in the same way as in Section \ref{section:mourre}, using $B^\sigma$ as a conjugate operator. Note that the ``perturbation'' $W^\sigma$ is simpler here than the one considered in Section \ref{section:Feshbach}, in that it only consists of the sum of a creation and an annihilation operator. However, some exponential decay in the electronic position variables $x_j$ has to be used in order to control the commutator of $W^\sigma$ with $B^\sigma$. (We do not present details.)
\end{proof}

\newpage

\section{List of notations}\label{appendix:notations}
\noindent \textbf{Hilbert spaces} \\
\begin{footnotesize}
\begin{align}
&\mathcal{H} = \mathrm{L}^2( \mathbb{R}^3 ) \otimes \mathcal{F}, \phantom{\frac{1}{2}} \\
&\mathcal{F} = \Gamma_s( \mathrm{L}^2( \mathbb{R}^3 \times \mathbb{Z}_2 )), \phantom{\frac{1}{2}} \\
&\mathcal{F_\sigma} = \Gamma_s( \mathrm{L}^2( \{ (k,\lambda) \in \mathbb{R}^3 \times \mathbb{Z}_2 , |k| \ge \sigma \} ) ), \phantom{\frac{1}{2}} \\
&\mathcal{F}^\sigma = \Gamma_s( \mathrm{L}^2( \{  (k,\lambda) \in \mathbb{R}^3 \times \mathbb{Z}_2 , |k| \le \sigma \} ) ). \phantom{\frac{1}{2}}
\end{align}
\end{footnotesize}
\noindent \textbf{Hamiltonians} \\
\begin{footnotesize}
\begin{align}
& H = \frac{1}{2} ( P - P_f - \alpha^{\frac{1}{2}} A )^2 + H_f,  \\
& H_\sigma = \frac{1}{2} ( P - P_f - \alpha^{\frac{1}{2}} A_\sigma )^2 + H_f \text{ (as an operator on } \mathcal{F} \text{)}, \\
& \phantom{H_\sigma } = K_\sigma \otimes \mathds{1} + \mathds{1} \otimes \big ( \frac{1}{2} P_f^2 + H_f ) - \nabla K_\sigma \otimes P_f \text{ (as an operator on } \mathcal{F_\sigma} \otimes \mathcal{F}^\sigma \text{)}, \\
& \nabla H_\sigma = P - P_f - \alpha^{\frac{1}{2}}A_\sigma, \phantom{\frac{1}{2}} \\
& K_\sigma = H_\sigma |_{\mathcal{F_\sigma}}, \quad \nabla K_\sigma = \nabla H_\sigma |_{\mathcal{F_\sigma}}, \phantom{\frac{1}{2}} \\
& U_\sigma = H - H_\sigma, \phantom{\frac{1}{2}} \\
& T_\sigma = K_\sigma \otimes \mathds{1} + \mathds{1} \otimes \big ( \frac{1}{2} P_f^2 + H_f ) - \nabla E_\sigma \otimes P_f, \\
& W_\sigma = H - T_\sigma = U_\sigma - ( \nabla K_\sigma - \nabla E_\sigma ) \otimes P_f, \phantom{\frac{1}{2}} \\
& H_\chi = T_\sigma + \chi W_\sigma \chi, \quad H_{\bar\chi} = T_\sigma + \bar\chi W_\sigma \bar\chi, \phantom{\frac{1}{2}} \\
& H_{\bar\chi}^1 = T_\sigma - \bar\chi ( \nabla K_\sigma - \nabla E_\sigma ) \otimes P_f \bar \chi, \phantom{\frac{1}{2}} \\
& F = F_\chi( H - \lambda ) |_{\mathrm{Ran}(P_\sigma\otimes\mathds{1})} \phantom{\frac{1}{2}} \\
& \phantom{F} = E_\sigma - \lambda + \mathds{1} \otimes \big ( \frac{1}{2} P_f^2 + H_f \big ) - \nabla E_\sigma \otimes P_f + \chi U_\sigma \chi - \chi W_\sigma \bar \chi \big [ H_{\bar \chi} - \lambda \big ]^{-1} \bar \chi W_\sigma \chi, \phantom{\frac{1}{2}} \\
& F_0 = E_\sigma- \lambda + \mathds{1} \otimes \big ( \frac{1}{2} P_f^2 + H_f \big ) - \nabla E_\sigma \otimes P_f, \phantom{\frac{1}{2}} \\
& W_1= \chi U_\sigma \chi, \phantom{\frac{1}{2}} \\
& W_2= - \chi W_\sigma \bar \chi \big [ H_{\bar \chi} - \lambda \big ]^{-1} \bar \chi W_\sigma \chi, \phantom{\frac{1}{2}} \\
& \tilde F = F + \lambda - E_\sigma, \quad \tilde F_0  = F_0 + \lambda - E_\sigma. \phantom{\frac{1}{2}}
\end{align}
\end{footnotesize}

\noindent \textbf{Conjugate operators} \\
\begin{footnotesize}
\begin{align}
& B = \d \Gamma ( b ), \quad b = \frac{ \i }{2} ( k \cdot \nabla_k + \nabla_k \cdot k ), \phantom{\frac{1}{2}} \\
& B^\sigma =  \d \Gamma( b^\sigma ), \quad b^\sigma = \kappa^\sigma b \kappa^\sigma. \phantom{\frac{1}{2}}
\end{align}
\end{footnotesize}

\noindent \textbf{Intervals} \\
\begin{footnotesize}
\begin{align}
& E = \inf \sigma( H ), \quad E_\sigma = \inf \sigma( H_\sigma ), \phantom{\frac{1}{2}} \\
& J_\sigma^> = E + [ \sigma , 2 \sigma ] \text{ (for } \sigma \ge \mathrm{C}_0 \alpha^{\frac{1}{2}} \text{)}, \phantom{\frac{1}{2}} \\
& J_\sigma^< = E + [ 11 \rho \sigma / 128 , 13 \rho \sigma / 128 ] \text{ (for } \sigma \le \mathrm{C}'_0 \alpha^{\frac{1}{2}} \text{)}, \phantom{\frac{1}{2}} \\
& \rho: \text{ fixed parameter such that } 0 < \rho < 1 \text{ and } \mathrm{Gap}( K_\sigma ) \ge \rho \sigma, \phantom{\frac{1}{2}} \\
& \Delta_\sigma = [ - \rho \sigma / 128 , \rho \sigma / 128 ], \phantom{\frac{1}{2}} \\
& \Delta'_\sigma = [ \rho \sigma / 16 , \rho \sigma / 8 ], \phantom{\frac{1}{2}}
\end{align}
\end{footnotesize}

\begin{figure}[htbp]
\centering
\resizebox{0.9\textwidth}{!}{\includegraphics{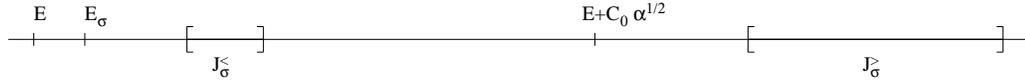}}
\caption{\textbf{The intervals $J_\sigma^<$ and $J_\sigma^>$}} \label{Intervals}
\end{figure}

\noindent \textbf{Functions} \\
\begin{footnotesize}
\begin{align}
& \kappa^\Lambda \in \mathrm{C}_0^\infty( \{ k , |k| \le \Lambda \} ; [0,1] )\ \mbox{and}\ \kappa^\Lambda =1\ \mbox{on}\ \{ k , |k| \le 3 \Lambda / 4 \}, \phantom{\frac{1}{2}} \\
& f_\sigma \in \mathrm{C}_0^\infty( [ 3 \rho \sigma / 64 ; 9 \rho \sigma / 64 ] ; [0,1] )\ \mbox{and}\ f_\sigma =1\ \mbox{on}\ \Delta'_\sigma, \phantom{\frac{1}{2}} \\
& \tilde f_\sigma : \text{ almost analytic extension of } f_\sigma. \phantom{\frac{1}{2}}
\end{align}
\end{footnotesize}

\noindent \textbf{(Almost) projections} \\
\begin{footnotesize}
\begin{align}
& P_\sigma = \mathds{1}_{ \{ E_\sigma \} }( K_\sigma ), \quad \bar P_\sigma = \mathds{1} - P_\sigma, \phantom{\frac{1}{2}} \\
& \chi_f^\sigma = \kappa^{\rho\sigma}( H_f ), \quad \bar \chi_f^\sigma = \sqrt{ \mathds{1} - (\chi_f^\sigma)^2 }, \phantom{\frac{1}{2}} \\
& \chi = P_\sigma \otimes \chi_f^\sigma, \quad \bar \chi = P_\sigma \otimes \bar \chi_f^\sigma + \bar P_\sigma \otimes \mathds{1}. \phantom{\frac{1}{2}}
\end{align}
\end{footnotesize}

\bibliographystyle{amsalpha}

\end{document}